\documentclass[10pt, Conference, letterpaper]{IEEEtran}
\IEEEoverridecommandlockouts
\usepackage{cite}
\usepackage{amsmath,amssymb,amsfonts}
\usepackage{pifont}
\usepackage{booktabs}
\usepackage{array} 

\usepackage{algorithmic}
\usepackage{graphicx}
\usepackage{textcomp}
\usepackage{xcolor}
\usepackage{amsthm}

\usepackage{float}  
\usepackage{subfigure}

\usepackage{hyperref}
\usepackage[ruled, vlined, linesnumbered]{algorithm2e}
\usepackage{enumitem}
\usepackage{balance}

\usepackage{fancyhdr}

\newtheorem{definition}{Definition}%[section]

\newtheorem{theorem}{Theorem}%[section]
\newtheorem{lemma}{Lemma}%[section]

\newcommand{\hollowstar}{\text{\ding{73}}}

\newcommand{\nnreals}{\mathbb{R}_{\ge0}}

\renewcommand{\epsilon}{\varepsilon}

\newcommand{\calG}{\mathcal{G}}

\newcommand{\calM}{\mathcal{M}\hspace{0.1mm}}

\newcommand{\calY}{\mathcal{Y}}

\newcommand{\E}{\operatorname{\mathbb{E}}}
\newcommand{\V}{\operatorname{Var}}

\newcommand{\bma}{\mathbf{a}}

\def\BibTeX{{\rm B\kern-.05em{\sc i\kern-.025em b}\kern-.08em
    T\kern-.1667em\lower.7ex\hbox{E}\kern-.125emX}}

\begin{document}

\title{Graph Analysis in Decentralized Online Social Networks with Fine-Grained Privacy Protection}

\author{\IEEEauthorblockN{Lele Zheng\IEEEauthorrefmark{2}, 
Bowen Deng\IEEEauthorrefmark{2},
Tao Zhang\IEEEauthorrefmark{2}, Yulong Shen\IEEEauthorrefmark{2} and Yang Cao\IEEEauthorrefmark{4}}
\\

%\IEEEauthorblockA{\IEEEauthorrefmark{2}xxx, China, \href{mailto:llzheng@stu.xidian.edu.cn}{\{xxx, xxx\}@xxx}\\
%\IEEEauthorrefmark{4}xxx, China,  \href{mailto:xxx@xxx}{xxx@xxx}}
}

\maketitle

\begin{abstract}
Graph analysts cannot directly obtain the global structure in decentralized social networks, and analyzing such a network requires collecting local views of the social graph from individual users. Since the edges between users may reveal sensitive social interactions in the local view, applying differential privacy in the data collection process is often desirable, which provides strong and rigorous privacy guarantees. In practical decentralized social graphs, different edges have different privacy requirements due to the distinct sensitivity levels. However, the existing differentially private analysis of social graphs provide the same protection for all edges. To address this issue, this work proposes a fine-grained privacy notion as well as novel algorithms for private graph analysis. We first design a fine-grained relationship differential privacy (FGR-DP) notion for social graph analysis, which enforces different protections for the edges with distinct privacy requirements. Then, we design algorithms for triangle counting and $k$-stars counting, respectively, which can accurately estimate subgraph counts given fine-grained protection for social edges. We also analyze upper bounds on the estimation error, including $k$-stars and triangle counts, and show their superior performance compared with the state-of-the-arts. Finally, we perform extensive experiments on two real social graph datasets and demonstrate that the proposed mechanisms satisfying FGR-DP have better utility than the state-of-the-art mechanisms due to the finer-grained protection.
\end{abstract}


\section{Introduction}\label{Introduction}
,
Decentralized Online Social Networks (DOSNs) \cite{DOSNS,paul2014survey,dosn2010} have recently received increasing attention because of the more control given to them over their shared content. As one of the most basic data patterns in DOSNs, social graph contain a wealth of valuable knowledge to uncover and thus analyzing social graph becomes a hot topic in recent years. As one of the most fundamental tasks in social graph analysis, counting subgraphs (e.g., triangles, stars) can be used to analyze the connection patterns in various social graphs, where the whole graph consists of different users’ local views. These subgraphs play an essential role in the social recommendation and constructing graph models.

A distcintive characteristic of DOSNs is that the data analysts often cannot obtain the entire social graph. Instead, analysts need to communicate with individual participants of the network, each with a limited local view of the entire social graph. However, users’ local views contain a lot of sensitive information since the edges usually reflect sensitive social interactions between individuals. Therefore, the analysis of the social graph must ensure strict privacy guarantees. Differential privacy \cite{Laplace}, as a privacy protection model with rigorous privacy guarantee, has become the gold standard for privacy analysis. However, the assumption that the server must be trusted makes it unsuitable for distributed online social networks. Local differential privacy (LDP) \cite{LDP} is a variant of differential privacy that allows each user to perturb her graph metrics locally before sending them to data analysts. Thus, it can be applied to decentralized online social networks. Several works have demonstrated the potential of LDP for private graph analysis, such as degree distribution, triangle counting, $k$-start.

\setcounter{subfigure}{0}
\begin{figure}[t]
	\centering
	\subfigbottomskip=16pt 
	\subfigcapskip=-5pt 
	\subfigure[Previous privacy model]{
		\includegraphics[width=0.4\linewidth]{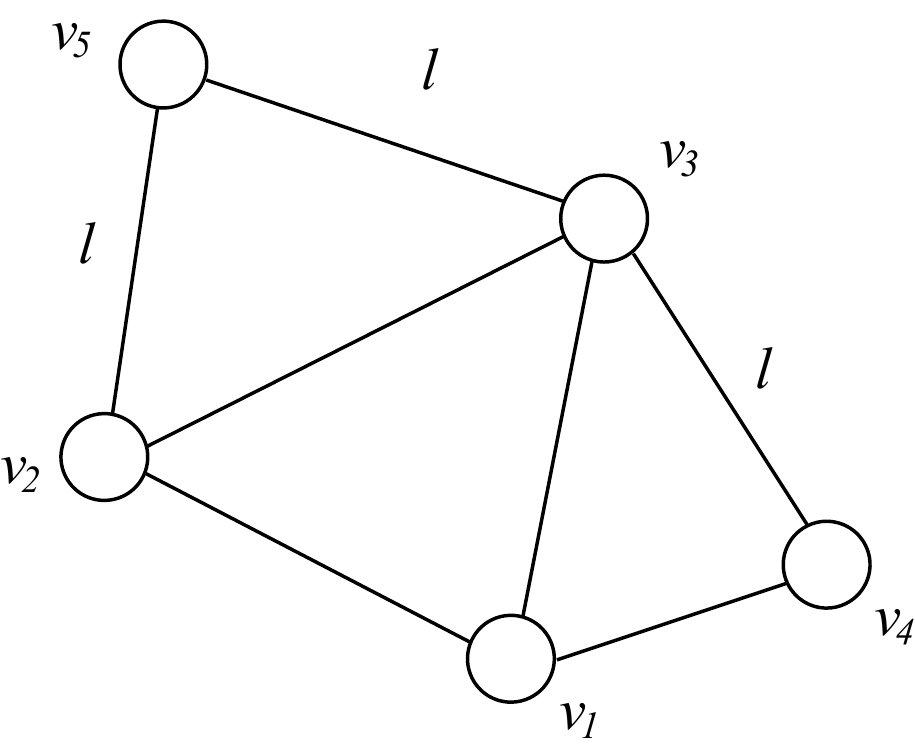}}
	\subfigure[Fine-grained privacy model]{
		\includegraphics[width=0.4\linewidth]{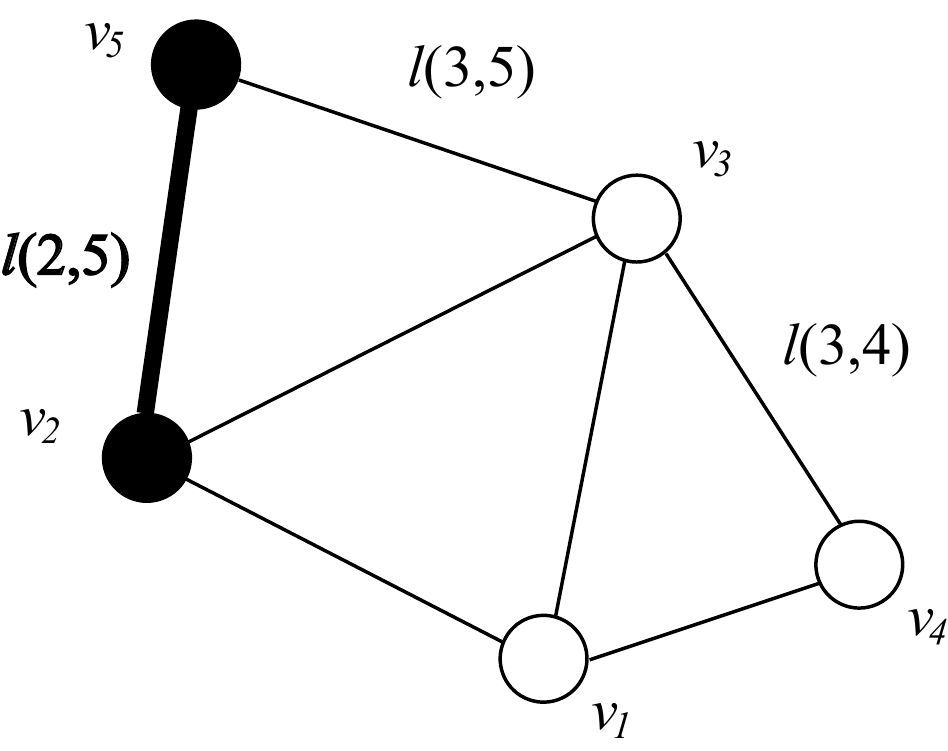}}
	\caption{(a) Previous works provide the same  privacy protection level $l$ for all social edges. (b) The proposed framework provides an appropriate privacy protection level for each edge, where the bold edge $\left \langle v_2, v_5 \right \rangle$ requires a higher privacy level.}
\end{figure}

The above differential privacy graph analysis mechanisms assume that all social edges are equally sensitive (controlled by the same privacy budget $\varepsilon$). Due to the uniform privacy budget, existing mechanisms would perturb the data in the same way (or add the same amount of noise) for all local views. However, in a real decentralized social network scenario, different edges may have different degrees of sensitivity and thus require different levels of privacy. For instance, social networks have a lot of different groups. Since the social interactions between the core nodes of the group often have a more decisive influence, they are more sensitive relative to the average users. The edges between these nodes require stronger privacy guarantees. A unified protection strategy will not only overprotect the unimportant edges of the social graph, reducing the utility of graph analysis but also cause issues such as insufficient protection of the core edges. Unfortunately, to our knowledge, fine-grained social edge protection is not considered by any existing differential privacy graph analysis mechanisms.

Motivated by the above observation, we consider differentially private analysis of social edges in decentralized social networks and assume that all input universes have multiple privacy levels represented by different values of privacy budgets. In practice, the privacy level of the edges between nodes can be classified by their influence. An edge between two core nodes can be classified as a strict privacy level, while an edge between two average nodes can be classified as a relaxed one. Since each possible edge $\left \langle v_i, v_j \right \rangle$ in graph $G$ has a privacy budget $\varepsilon_{l(i,j)}$ (edges with the same privacy level have the same privacy budget), the privacy budget of existing methods should be $\varepsilon = {\rm min}_{\left \langle v_i, v_j \right \rangle \in G}  \{\varepsilon_{l(i,j)}\}$ to satisfy the privacy requirements of all edges, which would provide overprotection for some inputs that do not require such strong privacy and lead to poorly estimated utility. What’s more, each subgraph is not independent in the social graph since different subgraphs may contain the same edge. There may also be complex interdependencies involving multiple people in the graph. For example, each node cannot estimate the triangle count of the local view because she cannot see the edges between other nodes; e.g., node $v_3$ cannot see an edge between $v_2$ and $v_5$ in Fig. 1.

This paper aims to consider fine-grained privacy protection for graph data. We first propose a privacy definition called FGR-DP in the local setting, which can provide different protections for edges with varying levels of privacy. We focus on triangle counts and $k$-stars counts - the most basic and valuable subgraphs counting tasks.
Specifically, our main contributions are summarized as follows:
\begin{enumerate}
    \item To satisfy the different privacy demands of different edges in decentralized social networks, we develop a practical privacy definition called fine-grained relationship DP (FGR-DP), which can provide fine-grained privacy protection for different edges of the social graph.
    \item Under fine-grained relationship DP, we propose an intuitive and efficient algorithm for $k$-stars count collection using the Laplace mechanism, which can achieve higher accuracy while satisfying the privacy requirement of edges.
    \item For triangle counting, a sophisticated two-phase algorithm with a solid privacy analysis is designed, where the node perturbs each neighboring edge independently according to the specific demands of privacy preservation, under the stringent FGR-DP notion. 
    \item We conduct extensive experiments over two real social graphs. The results show that the proposed technique consistently outperforms existing solutions in terms of result accuracy.
\end{enumerate}

The rest of the paper is organized as follows. Section \ref{Related Work} introduces the related work. Section \ref{Problem Statement and Preliminaries} describes local differential privacy and its application in graph analytics, as well as the system model. Section \ref{FG-RDP} presents the notion of FGR-DP. Section \ref{Subgraph Counting Mechanism} describes the proposed privacy-preserving subgraph counting algorithms in detail. Section \ref{Experiments} shows the experimental results. Finally, Section \ref{Conclusions} concludes the paper.

\section{Related Work}\label{Related Work}
As a gold standard, LDP has been considered by researchers for distributed graph analysis in recent years. Here we review some representative works, especially recent research with LDP of graph statistics. 

\textbf{Graph analysis with LDP.} Privacy-preserving subgraph counting plays an important role in graph analysis in decentralized social networks and has been addressed by several related works with different approaches \cite{sun2019analyzing}, \cite{liu2022collecting}, \cite{ye2020lf}, \cite{usenixcommunication}. Sun et al. \cite{sun2019analyzing} first consider this problem and propose DDP, a relaxation of LDP, which requires each user to consider the privacy of herself and her neighbors simultaneously, protecting the edge privacy of correlated data to a certain extent. Unfortunately, since DDP only hides one edge of the global, it provides a weak privacy guarantee. To address this challenge, Liu et al.\cite{liu2022collecting} define Edge-RLDP, which can provide a strong privacy guarantee when considering data correlation. Based on this, they propose a two-stage framework to achieve better estimation utility with strict privacy guarantees. Ye et al. \cite{ye2020lf} provide a generic graph metric estimation framework to support LDP graph analysis, called LF-GDPR, which simplifies developing a practical LDP solution for a graph analysis task by providing a complete solution for all LDP steps. Unfortunately, these LDP graph analysis methods assume that all edges are equally sensitive and cannot provide fine-grained protection for social graphs.

\textbf{Fine-grained LDP.}
The standard LDP assumes that all data are equally sensitive, resulting in excessive perturbation for some inputs and low utility. In reality, not all personal information should be treated equally  \cite{PLDPhistogram} \cite{PLDP} \cite{ULDP}  \cite{meanLDP}  \cite{IDLDP}. Nie et al.\cite{PLDPhistogram} consider the different demands of users and propose PLDP, which does not impose global privacy constraints on all users but instead follows each user's individual privacy requirements. Under PLDP, each participant can select the privacy level according to her preference. Murakami et al.\cite{ULDP} believe that some data may not require protection due to their inherent characteristics. For example, in a survey on exam cheating, "No" is naturally a non-sensitive response to this question. Therefore, they consider the inputs’ sensitivity level by directly classifying them as sensitive and non-sensitive. Gu et al. \cite{IDLDP} further demonstrate that different data have varying degrees of sensitivity. As a result, they present Input-Discriminative LDP (ID-LDP), a finer-grained variant of LDP for input data. However, none of these LDP variants are suitable for graph data analysis.

\section{Problem Statement and Preliminaries}\label{Problem Statement and Preliminaries}
\subsection{Notations}
\textbf{Graph.} 
 An undirected graph is defined as $G=(V,E)$, where $V$ is the set of nodes (i.e., users) and $E$ is the set of edges. Let $n$ be the number of nodes in $V$, and let $v_i \in V$ the $i$-th node; i.e., $V = \{v_1,v_2,\cdots,v_n\}$. An edge $\left\langle v_i, v_j\right\rangle \in E$ represents a relationship between nodes $v_i \in V$ and $v_j \in V$. The number of edges connected to a node is called the degree of the node. Let $d_{max}$ be the maximum degree (i.e., maximum number of edges connected to a node) in graph $G$.  A graph $G$ can be represented as a symmetric adjacency matrix $\mathbf{A}=(a_{i,j}\in \{ 0,1\}^{n \times n})$, where $a_{i,j}=1$ if and only if $\left\langle v_i, v_j\right\rangle \in E$ and otherwise $a_{i,j}=0$. The main symbols used in this paper are given in Table I.

\begin{definition}[\textbf{Neighboring graphs \cite{imola2021locally}}] 
    Given two graphs $G=(V, E)$ and $G'=(V', E')$, for any edge $\left\langle v_i, v_j\right\rangle \in E$, if $V' = V$ and $E'=E- \{\left\langle v_i, v_j\right\rangle\}$, then $G$ and $G'$ are neighboring graphs.
\end{definition}

\begin{definition}[\textbf{Local Differential Privacy (LDP)\cite{LDP}}]
    For a given $\varepsilon > 0$, a randomized algorithm $\calM$ satisfies $\varepsilon$-LDP if and only if for any pair of inputs $x, x'$ and any output $y$
    \begin{equation}
        \dfrac{Pr(\calM(x)=y)}{Pr(\calM(x')=y)} \leqslant e^\varepsilon
    \end{equation}
    where $\varepsilon$ is a parameter called privacy budget that controls the strength of privacy protection. A smaller $\varepsilon$ guarantees stronger privacy protection because the adversary has lower confidence when trying to distinguish any pair of inputs $x, x'$.
\end{definition}

The most widely employed mechanism for achieving differential privacy is Laplace mechanism.

\begin{definition}[\textbf{Laplace Mechanism\cite{Laplace}}] 
    let $f: G \rightarrow \mathcal{Y}$, the Laplace Mechanism is defined as 
\begin{equation}
    \mathcal{M}(G)=f(G)+  Lap(\Delta f/\varepsilon) 
\end{equation}
where $f(\cdot)$ in graph estimation is the subgraph count query, $\Delta f=\max\limits_{G, G'}|| f(G)-f(G') ||_1$ is the sensitivity. $Lap(\lambda)$ is a zero-mean Laplace distribution with scale $\lambda$, where $Lap(x|\lambda)=\frac{1}{2\lambda}exp(-\frac{|x|}{\lambda})$, and its variance is $(\Delta f/\varepsilon)^2$.
\end{definition}

\begin{table}\label{table-1}
\caption{List of Symbols}
\center
\begin{tabular}{l|l}\hline % 其中，tabular是表格内容的环境；c表示centering，即文本格式居中；c的个数代表列的个数
\toprule %[2pt]设置线宽     
Symbol & Description \\ %换行
\midrule %[2pt]  
$U$                   & the set of users\\
$n$                   & Number of users\\
$G=(V,E)$             & Graph with nodes(users) $V$ and edges $E$ \\
$\calG$               & Set of possible graphs with $n$ users\\
$v_i$                 & $i$-th user in $V$\\
$f_\bigtriangleup(G)$ & Number of triangles in $G$ \\
$f_\hollowstar(G)$    & Number of $k$-stars in $G$ \\
$L$                   & Number of privacy levels\\
$l(i,j)$              & Privacy level of edge $\left\langle v_i, v_j\right\rangle$\\
$\varepsilon_l$       & Privacy budget for privacy level $l$ \\
$d_{max}$, $\widetilde{d}_{max}$ & True vs. estimated maximum degree\\
$\mathbf{A}=(a_{i,j})$    & Adjacent matrix\\
$\bma_i$               & $i$-th row of $\mathbf{A}$ (i.e., Neighbors of $v_i$)\\  
\bottomrule %[2pt]     
\end{tabular}
\end{table}

\begin{definition}[\textbf{Randomized Response\cite{RR}}]
Random Response (RR)  can protect the sensitive Boolean responses of participating users in a survey. Specifically, each participant gives the true answer with probability $p$ and the opposite answer with probability $1-p$. To satisfy $\varepsilon$-LDP, the probability is selected as $p = \dfrac{e^\varepsilon}{1+e^{\varepsilon}}$.
\end{definition}

\begin{definition}[\textbf{$\varepsilon$-relationship DP \cite{imola2021locally}}] 
  Let $\epsilon \in \nnreals$. For \mbox{}$1 \leq i \leq n$, let $\calM_i$ be an obfuscated mechanism of user $u_i$ that takes $\bma_i$ as input. We say $(\calM_1, \cdots, \calM_n)$ provides \emph{$\varepsilon$-relationship DP} if for any two neighboring graphs $G, G' \in \calG$ that differ in one edge and  any $(\calY_1, \cdots, \calY_n) \in \mathrm{Range}(\calM_1) \times \cdots \times \mathrm{Range}(\calM_n)$, 
\begin{align}
  &\Pr[(\calM_1(\bma_1), \cdots, \calM_n(\bma_n)) = (\calY_1, \cdots, \calY_n)] \nonumber\\
  &\leq e^\epsilon \Pr[(\calM_1(\bma'_1), \cdots, \calM_n(\bma'_n)) = (\calY_1, \cdots, \calY_n)],
\label{eq: relation_LDP}
\end{align}

  where $\bma_i$ (resp. $\bma_i'$) $\in \{0,1\}^n$ is the $i$-th row of the adjacency matrix of graph $G$ (resp. $G'$).
\end{definition}
As our work focuses on fine-grained privacy protection of edges for decentralized social graphs, we further define fine-grained relationship DP in section \ref{FG-RDP}.

\begin{theorem}[\textbf{Sequential Composition of LDP\cite{mcsherry2009privacy}}]\label{Sequential Composition}
    If randomized algorithm $\mathcal{M}_i: \mathcal{X} \rightarrow \mathrm{Range}(\mathcal{M}_i)$ satisfies $\varepsilon_i$-LDP for $i=1, 2, \cdots, n$, then the sequential combination of these algorithms $\mathcal{M}_i (1\leq i \leq n)$ satisfies $(\sum \varepsilon_i)$-LDP.
\end{theorem}

\subsection{Problem Statement}
\textbf{System Model.} 
Our system model involves a data analyst and $n$ individual users $U = \{u_1, u_2, \cdots, u_n\}$. Each user has a limited local view $G_i=(V_i, E_i)$ of the global social graph $G$ and individually obfuscates sensitive data employing a random perturbation algorithm before sending it to the data analyst.
Then, the data analyst combines information from different users to evaluate the whole social network properties. We consider two types of most basic and useful subgraph counts, and one is the triangle counts $f_\bigtriangleup(G)=|\{ \forall (i<j<k)|(v_i,v_j,v_k\in V)\wedge(\left\langle v_i, v_j\right\rangle,\left\langle v_j, v_k\right\rangle,\left\langle v_j, v_k\right\rangle\in E)\}|$, where a triangle is a set of three nodes with three edges; the other is the $k$-stars counts $f_\hollowstar(G)=|\{ \forall (i,j_1<j_2<...<j_k)|(v_i,v_{j_1}, ... ,v_{j_k}\in V)\wedge(\left\langle v_i, v_{j_1}\right\rangle, ... ,\left\langle v_i, v_{j_k}\right\rangle\in E)\}|$, where a $k$-stars consists of a central node connected to $k$ other nodes. Counting them is an essential task in analyzing graph properties. For example, the data analyst can calculate clustering coefficients based on these two types of subgraph counts. We employ differential privacy to protect users' sensitive social interaction information. Assume there are $L$ privacy levels, and the privacy level of edge $\left\langle v_i, v_j\right\rangle$ is $l(i,j)$. Though the whole social graph $G$ can be large, the number of privacy levels determined by influence is usually small in practice (and usually only two levels). For convenience, we define the privacy budget of any edge $\left\langle v_i, v_j\right\rangle$ as $\varepsilon_{l(i,j)}$.

\section{Fine-grained Relationship DP}\label{FG-RDP}

In this section, we introduce a novel privacy concept, fine-grained relationship DP, which protects the existence of arbitrary edges in LDP graphs. In fine-grained relationship DP, the sensitivity of an edge $\left\langle v_i, v_j\right\rangle$ is determined by its two endpoints $v_i$ and $v_j$. Meanwhile, we formally analyze the requirements for implementing the FGR-DP. 

\subsection{Definition}
LDP defines privacy as the highest level of indistinguishability between any two adjacent graph data. In real-world applications, different edges may have different privacy levels. As a result, the indistinguishability requirements between different adjacent graphs may be distinct. However, LDP cannot provide such fine-grained privacy protection because its definition is based on the worst-case scenario. This uniform definition would lead to numerous drawbacks, such as data overprotection and low data utility. Intuitively, since less noise can be added to low-sensitive edges, providing fine-grained privacy protection for different adjacency graphs can improve the utility of subgraph counting. We describe the new notion of fine-grained relationship DP as follows.

\begin{definition}[Fine-grained Relationship DP] \label{def:Fine_DP} 
For \mbox{$1 \leq i \leq n$}, let $\calM_i$ be an obfuscated mechanism of user $v_i$ that
takes $\bma_i$ as input. We say $(\calM_1, \cdots, \calM_n)$ provides 
\emph{fine-grained relationship DP}
if for any two neighboring graphs $G, G^\prime \in \calG$ that differ in edge 
$\left\langle v_i, v_j\right\rangle$ and any $(\calY_1, \cdots, \calY_n) \in \mathrm{Range}(\calM_1) \times \cdots \times \mathrm{Range}(\calM_n)$, 
\begin{align}
  &\Pr[(\calM_1(\bma_1), \cdots, \calM_n(\bma_n)) = (\calY_1, \cdots, \calY_n)] \nonumber\\
  &\leq e^{\varepsilon_{l(i,j)}} \Pr[(\calM_1(\bma'_1), \cdots, \calM_n(\bma'_n)) = (\calY_1, \cdots, \calY_n)],
\end{align}
  where $\varepsilon_{l(i,j)}$ denotes the privacy budget of edge $\left \langle v_i, v_j \right \rangle$.
\end{definition}

We assume that each edge $\left \langle v_i, v_j \right \rangle$ has a specific privacy level $l(i,j)$, and its corresponding privacy budget is $\varepsilon_{l(i,j)}$.
Intuitively, in Definition~\ref{def:Fine_DP}, the existence of each edge $\left \langle v_i, v_j \right \rangle$ is protected by $\varepsilon_{l(i,j)}$-LDP.
Theoretically, FGR-DP can be employed to perform a variety of graph data analysis tasks. This paper focuses on the privacy-preserving subgraph counting algorithm, under FGR-DP.

\textbf{Example.} Consider a scenario in which a data analyst collects subgraph (triangles and $k$-stars) counts from a decentralized social network to discover the clustering coefficient of the whole social graph $G=(V, E)$. The analyst needs to interact with $n$ independent users (nodes) and ask each user to return obfuscated output from the local social graph, where the privacy level $l(i, j)$ of the edge $\left \langle v_i, v_j \right \rangle$ is jointly determined by the influence of $v_i$ and $v_j$. Since the relationship $\left \langle v_i, v_j \right \rangle$ between two influential nodes (e.g., a celebrity and a government official) is more sensitive than the other relationship $\left \langle v_{i^\prime}, v_{j^\prime} \right \rangle$, the privacy budget $\varepsilon_{l(i,j)} < \varepsilon_{l(i^\prime,j^\prime)}$, where a smaller $\varepsilon$ indicates a higher privacy level that requires stronger privacy protection. Under FGR-DP, the edges with low sensitivity only need a small amount of noise, and such edges constitute the majority (usually more than 90\%) in decentralized social graphs.

\subsection{Implementation.}
In the process of implementation, the privacy level should satisfy the following properties:
\begin{lemma}[Symmetry]
For any $v_i$, $v_j\in V$ and $i\neq j$,
\begin{equation}
    l(i,j)=l(j,i).
\end{equation}
\end{lemma}
\begin{lemma}[Transferability]
For any $v_i$, $v_j$, $v_{i^\prime}$,$v_{j^\prime}$,$v_{i^{\prime\prime}}$,$v_{j^{\prime\prime}} \in V$, if $l(i,j)<l(i^\prime,j^\prime)$ and $l(i^\prime,j^\prime)<l(i^{\prime\prime},j^{\prime\prime})$, then
\begin{equation}
    l(i,j)<l(i^{\prime\prime},j^{\prime\prime})
\end{equation}
\end{lemma}

\begin{lemma}[Ordering]
For any $v_i$, $v_j$, $v_{i^\prime}$, $v_{j^\prime}\in V$, if $l(i,j)<l(i^\prime,j^\prime)$, then
\begin{equation}
    \varepsilon_{l(i,j)}<\varepsilon_{l(i^\prime,j^\prime)}
\end{equation}
\end{lemma}

We assume that the privacy level set is $\{1,2,...,L\}$, where $L\in \mathbb{N}^+$. The number of edges grows $O(n^2)$ with the nodes in the social graph. In this paper, user $v_i$ formulates her privacy level based on their neighboring relationships as follows:
\begin{equation}
    l(v_i) = \min_{1\leq j\leq n,j\neq i} \{l(i,j)\}.
\end{equation}

Straightforwardly, each user sets her privacy level to the strictest of the neighboring relationships. Before uploading the local view, each user perturbs all adjacent edges with $l(v_i)$ privacy level. Therefore, the edge $\left \langle v_i, v_j \right \rangle$ is actually protected by the privacy level of $l^\prime(i,j)=\max\{l(v_i),l(v_j)\}$. It’s easy to get
\begin{equation*}
    l^\prime(i,j)=\max\{l(v_i),l(v_j)\}\leq \max\{l(i,j),l(i,j)\} = l(i,j).
\end{equation*}

Thus, each edge can usually obtain stronger (at least no worse) privacy protection than the required level.

\section{Subgraph Counting Mechanism} \label{Subgraph Counting Mechanism}

In this section, we first consider that there are two different privacy levels in privacy-preserving social graph analysis, i.e., $L=2$. Intuitively, we classify the edges between nodes as high-sensitive (relationships between core nodes) and low-sensitive (relationships between ordinary nodes). To address this challenge, we propose two high-precision privacy-preserving algorithms to obtain unbiased estimates for $k$-stars counting and triangle counting, respectively. The privacy analysis reveals that the proposed algorithms satisfy fine-grained relationship differential privacy. Finally, we show that our algorithms can be naturally extended to multi-level $(L \geq 3)$ privacy-preserving graph data analysis. The error analysis demonstrates the advantage of our algorithms over existing methods in terms of accuracy.

\subsection{$k$-stars Counting}
Algorithm~\ref{alg:our-star} shows how the data analyst estimates the $k$-stars counts of the whole graph.
It takes three inputs-the whole graph $G$ (represented as neighbor lists $\bma_1, \cdots, \bma_n$), the privacy budget $\varepsilon_{l(i,j)}, l(i,j) \in \{1, 2\} $, and the estimated maximum degree $\widetilde{d}_{max}$, and returns an estimate of $k$-stars counts under FGR-DP. In Line 1, the data analyst first calculates the global sensitivity $\Delta f_\hollowstar$, the maximum number of $k$-stars increases by adding an edge. Then each node $v_i$ calculates its own privacy level $l(v_i)$ based on Eq.(9) and clips the adjacent edges to at most $\widetilde{d}_{max}$ (Line 3-4). Further, each node counts the number of $k$-stars $r_i$ in the local view (Line 5). After getting the $r_i$, they add an appropriate amount of Laplace noise to $r_i$ (Line 6), which is based on $l(v_i)$ (obtained from Line 3), and then submit the obfuscated values to the data analyst (Line 7). Finally, the data analyst applies aggregation to estimate the number of $k$-stars of the whole graph (Line 8).

\textit{Global sensitivity.} To satisfy FGR-DP, each node needs to add Laplace noise to the local $k$-stars count. Therefore, it is the first priority to calculate the global sensitivity (Line 1). However, the degree of each node is sensitive information because it can reveal the existence of edges. In decentralized social graphs, no one knows the maximum degree $d_{max}$ of the whole graph. This paper adopts the estimated maximum degree $\widetilde{d}_{max}$ to replace the true maximum degree $d_{max}$ based on the literature \cite{imola2021locally}. In essence, they first draw a very small privacy budget to estimate the global maximum degree $\widetilde{d}_{max}$. This approach is feasible because the estimation result is larger than the degree of most nodes.

\textit{Clip.} After calculating the privacy level, each node $v_i$ needs to clip the adjacent edges to control the local sensitivity because possible negative noise can result in $\widetilde{d}_{max}<d_{max}$. Specifically, if the degree $d_i$ of node $v_i$ is less than $\widetilde{d}_{max}$, then its adjacent edges will be preserved; otherwise, node $v_i$ will clip its adjacent edges so that the degree $d_i$ is equal to $\widetilde{d}_{max}$. Since an edge can affect the $k$-stars counts of two nodes simultaneously, each node needs to halve the privacy budget based on Theorem 1.

\begin{algorithm}[t]
\caption{Fine-grained LDP for $k$-stars counting.}\label{alg:our-star}
\SetKwInOut{Input}{Input}
\SetKwInOut{Output}{Output}
\Input{Graph $G$ represented as neighbor lists $\bma_1, \cdots, \bma_n \in  \{0, 1\}^n;$
         \\ Privacy budgets $\varepsilon_{l(i,j)}, l(i,j) \in \{1, 2\};$ 
          \\Estimated maximum \mbox{degree $\tilde{d}_{max}$};}
\Output{The estimation result of $f_\hollowstar(G)$.}
    Calculate the global sensitivity $\Delta f_\hollowstar = \dbinom{\tilde{d}_{max}}{k-1}$\\
    \For{each node $v_i$ }
    {
            Calculate privacy level $l(v_i)$ according to Eq.(9); \\
    	\textit{Clip} the neighboring edges to at most $\tilde{d}_{max}$;\\
    	$r_i=\dbinom{d_i}{k}$;\\
    	$\hat{r}_i=r_i+{\rm Lap}(\frac{D}{\varepsilon_{l(v_i)}/2})$;\\
    	Upload $\hat{r}_i$;
    }
    \Return $\sum_{i=1}^{n}\hat{r}_i$
\end{algorithm}

\begin{theorem}\label{thm:star}
    In Algorithm~\ref{alg:our-star}, the existence of any edge $\left \langle v_i, v_j \right \rangle$ is protected by $\varepsilon_{l(i,j)}$-{\rm LDP}, where $l(i,j)\in\{1,2\}$  denotes the privacy level of edge $\left \langle v_i, v_j \right \rangle$.
\end{theorem}

\begin{proof}
   For any node $v_i$, when it adds an edge, the number of $k$-stars increases by $\dbinom{d_i}{k-1}$, and when it removes an edge, the number of $k$-stars decreases by $\dbinom{d_i-1}{k-1}$. Therefore, for the $k$-stars counting algorithm, the global sensitivity of the whole graph is $\dbinom{\widetilde{d}_{max}}{k-1}$ after clipping.

   Without loss of generality, we assume that $i < j$. Based on the Laplace mechanism (Definition 3) and the combination theorem (Theorem 1), we can easily obtain the privacy budget consumed by edge $\left \langle v_i, v_j \right \rangle$ is $\frac{\varepsilon_{l(v_i)}}{2}+\frac{\varepsilon_{l(j)}}{2}\leq\varepsilon_{l(i,j)}$. Thus, the existence of any edge $\left \langle v_i, v_j \right \rangle$ is protected by $\varepsilon_{l(i,j)}$-{\rm LDP}. 
\end{proof}

\begin{theorem}\label{thm:star-var}

    We refer to the $k$-stars counting algorithm by the function $KS(\cdot)$. Let $n_1, n_2\geq0$ $(n_1+n_2=n)$ be the number of nodes with privacy levels 1 and 2, respectively. For given $0<\varepsilon_1<\varepsilon_2$ (corresponding to the privacy level), $\tilde{d}_{max}\geq d_{max}$, 
    $KS(G,\varepsilon_1,\varepsilon_2)$ is an unbiased estimation of $f_\hollowstar(G)$. Formally, we have
    \begin{equation}
        \E[KS(G,\varepsilon_1,\varepsilon_2)]=f_\hollowstar(G),
    \end{equation}
    and the variance
    \begin{equation}
        Var[KS(G,\varepsilon_1,\varepsilon_2)]
        = 4\dbinom{\tilde{d}_{max}}{k-1}^2\left(\frac{n_1}{\varepsilon_1^2}+\frac{n_2}{\varepsilon_2^2}\right).
    \end{equation}
\end{theorem}
    
\begin{proof}
    For the $k$-stars counting algorithm, each node independently adds Laplacian noise to the local count before sending the obfuscated value to the data analyst. The data analyst sums up all the uploaded values to get the aggregated result. Thus, according to expectation additivity, $KS(G,\varepsilon_1,\varepsilon_2)$ is an unbiased estimation of $f_\hollowstar(G)$ since the expectation of Laplace noise is zero.
    
    Based on Definition 3, we can derive the variance as
\begin{equation}
\begin{aligned}
     &Var[KS(G,\varepsilon_1,\varepsilon_2)]\\ &= n_1 Var\left(Lap\left(\Delta f_\hollowstar /\frac{\varepsilon_1}{2}\right)\right)+n_2 Var\left(Lap\left(\Delta f_\hollowstar /\frac{\varepsilon_2}{2}\right)\right)\\
     &=4n_1\left(\dbinom{\tilde{d}_{max}}{k-1}/ \varepsilon_1\right)^2 + 4n_2\left(\dbinom{\tilde{d}_{max}}{k-1}/ \varepsilon_2 \right)^2 \\
     &=4\dbinom{\tilde{d}_{max}}{k-1}^2\left(\frac{n_1}{\varepsilon_1^2}+\frac{n_2}{\varepsilon_2^2}\right)
\end{aligned}
\end{equation}

\end{proof}

\subsection{Triangle Counting}
In decentralized social graphs, for any three nodes $v_i,v_j$, and $v_k$, node $v_i$ only knows the existence of two neighboring edges $\langle v_i,v_j\rangle$,  $\langle v_i,v_k\rangle$ and lacks the knowledge of the third edge $\langle v_j,v_k\rangle$, e.g., node $v_3$ cannot see the edge between $v_2$ and $v_5$ in Fig. 1. Therefore, the user cannot directly count the number of triangles in the local view. To solve this challenge, we require a two-round interaction algorithm in triangle counting. In the first round, each user uploads the obfuscated values of the adjacent edges. The data analyst integrates the received data into an adjacency matrix $\textbf{A}$ and sends it back to users. User $v_i$ can then observe a noisy edge $\langle v_j,v_k\rangle$ in the adjacent matrix and count the number of noisy triangles formed by  $(v_i, v_j, v_k)$. After this round, each user can obtain an unbiased estimate of the triangle count for their local view. In the second round, each user adds Laplace noise to the unbiased estimate according to the corresponding privacy budget and submits the noisy triangle count to the data analysts. The data analyst aggregates the values uploaded by all users to get an estimate of the triangle counting. We next show the detailed procedure of triangle counting under fine-grained privacy protection.

Algorithm~\ref{alg:our-sanjiao} describes how the data analyst estimates the triangle counts of a given social graph according to the privacy budget $\varepsilon_{l(i, j)}$ and the estimated maximum degree $\tilde{d}_{max}$. It first computes $n_1$, the number of nodes with privacy level is 1 based on the $l(v_i)$ (Line 1). Due to the two-round interaction, in Line 2 the privacy budget $\varepsilon_{l(i, j)}$ is divided into $\alpha\varepsilon_{l(i, j)}$ and $(1 -\alpha)\varepsilon_{l(i, j)}$ according to the privacy combination theorem (see Theorem 1 for details), and then $\alpha$ is broadcast to each node (Line 3). In Line 4, the data analyst reorders all nodes according to their privacy level to ensure that lower-order nodes have higher privacy levels than higher-order nodes, i.e., for any $v_i, v_j$, if $i < j$, then $l(v_i)\leq l(v_j)$. This operation can effectively reduce statistical errors since low-order (less privacy budget) nodes only need to upload a small number of relationships.

Round $1$ (Lines 5-10): Each node $v_i$ perturbs its adjacency bit vector independently according to the privacy level (Lines 6-9). For each bit to perturb, it adopts RR with privacy budget $\alpha\varepsilon_{l(i, j)}$ (the perturbation probability is calculated in Line 5). Then, each node uploads the perturbed adjacency vector $R_i$ based on its order, i.e., the $i$-th row of the lower triangle part of the adjacency matrix. As shown in Fig. 2, higher-order nodes upload more relationships. Finally, the data analyst consolidates uploaded vectors into an obfuscated adjacent matrix $\widetilde{A}$ and sends it to each node (Line 10).

\begin{figure}[!t]
	\centering
	\includegraphics[width=2.5in]{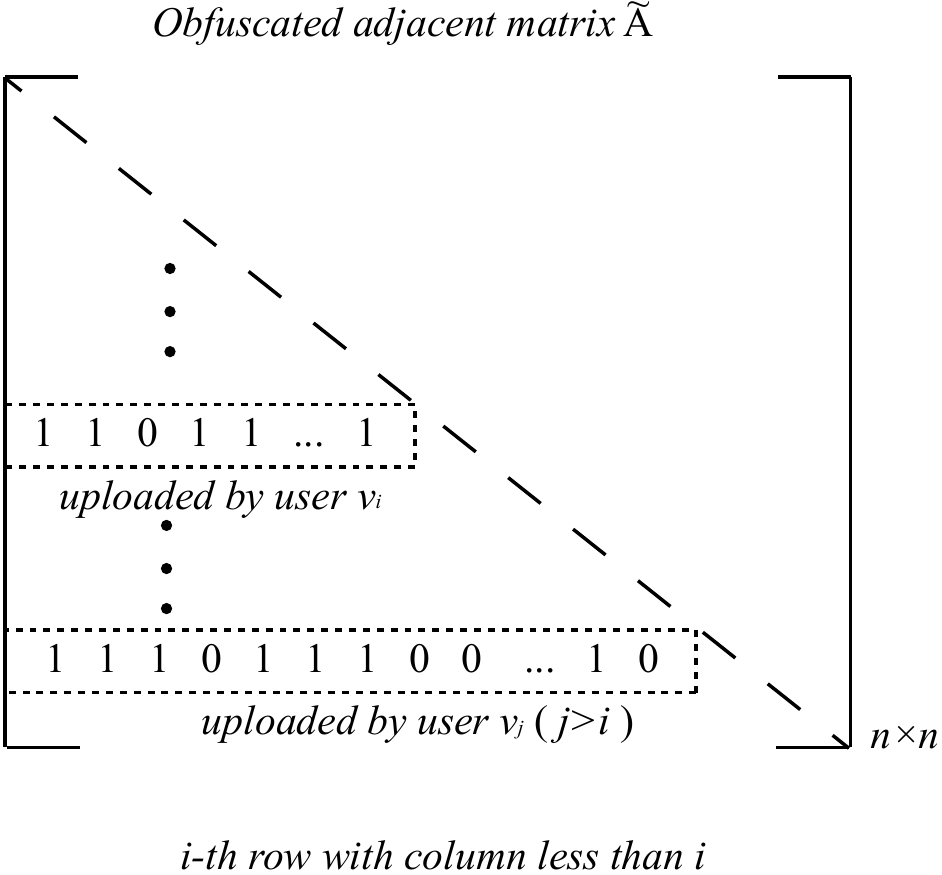}
	\caption{The obfuscated adjacent matrix (the dashed boxes are the parts uploaded by users $v_i$ and $v_j$).}
      \label{fig: matrix}
\end{figure}

Round $2$ (Lines 10-27):
We first calculate the global sensitivity $\Delta f_\bigtriangleup$ and perform clipping as in $k$-stars counting. According to the reordering result, each node only counts triangles formed by nodes that are higher in order than itself to avoid double counting. Node $v_i$ obtains an unbiased estimate $\widetilde{w}_i$ of the local triangle counts $|\{ \forall (i<j<k)|(v_i,v_j,v_k\in V)\wedge(\langle v_i,v_j\rangle,\langle v_i,v_k\rangle,\langle v_j,v_k\rangle\in E)\}|$ from the adjacency matrix $\widetilde{A}$. The estimate $\widetilde{w}_i$ consists of two parts, $\widetilde{w}_i^{(1)}$ and $\widetilde{w}_i^{(2)}$, where $\widetilde{w}_i^{(l)}$ is an unbiased estimate of the triangle counts for which the privacy level of the third edge $\langle v_j, v_k\rangle$ is $l$. 
When privacy level $l(v_i) = 1$, node $v_i$ needs to calculate both $\widetilde{w}_i^{(1)}$ and $\widetilde{w}_i^{(2)}$ since the privacy level of the third edge may be 1 or 2. To get the correct estimate, it first calculates the 2-stars count $s_i^{(l)}$ and  the positive triangle count $t_i^{(l)}$ when the third edge $\left \langle v_j, v_k \right \rangle$ privacy level is $l$ (Lines 13-16), where $i<j<k$. Then node $v_i$ can obtain $\widetilde{w}_i^{(l)}$ by  $s_i^{(l)}$, $t_i^{(l)}$, and $p_l$ in Lines 17-18.
Furthermore, it perturbs the unbiased triangle counts $\widetilde{w}_i^{(1)} + \widetilde{w}_i^{(2)}$ by adding a Laplace noise with privacy budget $(1-\alpha)\varepsilon_1$ (Line 19). Finally, it uploads the perturbed result $\hat{w}_i$ to the data analyst.
When privacy level $l(v_i) = 2$, node $v_i$ only needs to count triangles consisting of nodes with privacy level 2, i.e., $\widetilde{w}_i^{(2)}$. Similarly, it obtains unbiased triangle counts $\widetilde{w}_i^{(2)}$ as above (Lines 23-25). Then node $v_i$ adds a Laplace noise with privacy budget $(1-\alpha)\varepsilon_2$ (Line 26) and sends $\hat{w}_i$ to the data analyst (Line 27).
Finally, in Line 28 the data analyst sums up the uploaded data to get the estimation result.

\begin{algorithm}[t]
    \label{alg:our}
	\caption{Fine-grained LDP for triangle counting.}\label{alg:our-sanjiao}
	\SetKwInOut{Input}{Input}
	\SetKwInOut{Output}{Output}
	\Input{Graph $G$ represented as neighbor lists $\bma_1, \cdots, \bma_n \in  \{0, 1\}^n$;\\
                Privacy budgets $\varepsilon_{l(i, j)}, l(i,j) \in \{1, 2\}$;\\
                Estimated maximum \mbox{degree $\tilde{d}_{max}$.}}
	\Output{The estimation result of $f_\bigtriangleup(G)$.}
	Let $n_1$ be the number of nodes with privacy level $l(v_i)=1$;\\
        Set parameter $\alpha$ for privacy budget allocation; \\
        Send $\alpha$ to each node;\\
        Reorder($G$);\\
        
	Calculate $p_1=\frac{e^{\alpha \varepsilon_1}}{e^{\alpha \varepsilon_1}+1}$, $p_2=\frac{e^{\alpha \varepsilon_1}}{e^{\alpha \varepsilon_2}+1}$;\\
	\For{$i = 1$ to $n_1$}
	{$R_i = (RR_{\alpha \varepsilon_1}(a_{i,1}), RR_{\alpha \varepsilon_1}(a_{i,2}), ..., RR_{\alpha \varepsilon_1}(a_{i,i-1}))$; }
	\For{$i = n_1+1$ to $n$}
	{$R_i = (RR_{\alpha \varepsilon_2}(a_{i,1}), RR_{\alpha \varepsilon_2}(a_{i,2}), ..., RR_{\alpha \varepsilon_2}(a_{i,i-1}))$; }
	Synthesize  $R_i$ ($1\leq i\leq n$) into $G^\prime$;\\
	\For{$i = 1$ to $n_1$}
	{
		Clip the neighboring edges of $\bma_i$ to at most $\tilde{d}_{max}$;\\
		$t_i^{(1)}=|\{(v_i,v_j,v_k): i<j<k, k\leq n_1, a_{i,j} = a_{i,k} = 1, \left \langle j,k \right \rangle\in G^\prime\}|$;\\
		$s_i^{(1)}=|\{(v_i,v_j,v_k): i<j<k, k\leq n_1, a_{i,j} = a_{i,k} = 1\}|$;\\
		$t_i^{(2)}=|\{(v_i,v_j,v_k): i<j<k, k>n_1,, a_{i,j} = a_{i,k} = 1, \left \langle j,k \right \rangle\in G^\prime\}|$;\\
		$s_i^{(2)}=|\{(v_i,v_j,v_k): i<j<k, k>n_1, a_{i,j} = a_{i,k} = 1\}|$;\\
		$\widetilde{w}_i^{(1)}=\frac{1}{2p_1-1}(t_i^{(1)}-(1-p_1) s_i^{(1)})$;\\
		$\widetilde{w}_i^{(2)}=\frac{1}{2p_2-1}(t_i^{(2)}-(1-p_2) s_i^{(2)})$
		
		$\hat{w}_i=\widetilde{w}_i^{(1)} + \widetilde{w}_i^{(2)}+{\rm Lap}(\frac{\tilde{d}_{max}/(2p_1-1)}{(1-\alpha)\varepsilon_1})$;\\
		Upload $\hat{w}_i$;
	}
	\For{$n_1+1$ to $n$}
	{
		Clip the neighboring edges of $\bma_i$ to at most $\tilde{d}_{max}$;\\
		$t_i^{(2)}=|\{(v_i,v_j,v_k): i<j<k, a_{i,j} = a_{i,k} = 1, \left \langle j,k \right \rangle\in G^\prime\}|$;\\
		$s_i^{(2)}=|\{(v_i,v_j,v_k): i<j<k, a_{i,j} = a_{i,k} = 1\}|$;\\
		$\widetilde{w}_i^{(2)}=\frac{1}{2p_2-1}(t_i^{(2)}-(1-p_2) s_i^{(2)})$\\
		$\hat{w}_i=\widetilde{w}_i^{(2)}+ {\rm Lap}\left(\frac{\tilde{d}_{max}/(2p_2-1)}{(1-\alpha)\varepsilon_2}\right)$;\\
		Upload $\hat{w}_i$;
	}
	\Return $\sum_{i=1}^{n}\hat{w}_i$.
	
\end{algorithm}

\begin{theorem}\label{thm:FGT}
    In Algorithm~\ref{alg:our-sanjiao}, the existence of any edge $\left \langle v_i, v_j \right \rangle$ is protected by $\varepsilon_{l(i,j)}$-{\rm LDP}, where $l(i, j)\in\{1,2\}$  denotes the privacy level of edge $\left \langle v_i, v_j \right \rangle$.
\end{theorem}
\begin{proof}
     We assume that $i < j$. Since there are two different privacy levels ($l(i, j)\in\{1,2\}$), we analyze them in two cases.
    
    Case 1: $l(i, j) = 1$. For any edge $\left \langle v_i, v_j \right \rangle$, it consumes privacy budget $\alpha\varepsilon_1$ in the process of generating the obfuscated adjacency matrix (Round 1). After clipping, the degree $d_i\leq\tilde{d}_{max}$ of each node $v_i$ in the whole graph. 
    Therefore, adding or removing edge $\left \langle v_i, v_j \right \rangle$ will lead to 
    $|\Delta (s^{(1)}+s^{(2)})|\leq\tilde{d}_{max}$ and 
    $|\Delta (t^{(1)}+t^{(2)})|\leq\tilde{d}_{max}$.
    It is obvious that $\Delta s^{(1)}$, $\Delta s^{(2)}$, $\Delta t^{(1)}$, and $\Delta t^{(2)}$ are non-positive or non-negative at the same time.
    We have the following conclusions:
    \begin{equation}
    \begin{aligned}
        &|\Delta(\widetilde{w}^{(1)}+\widetilde{w}^{(2)})|\\
        &=|\frac{1}{2p_1-1}(\Delta t^{(1)} - p_1\Delta s^{(1)})\\
        &+ \frac{1}{2p_2-1}(\Delta t^{(2)} -  p_2\Delta s^{(2)})|\\
        &=|\frac{1}{2p_1-1}\Delta t^{(1)}+\frac{1}{2p_2-1}\Delta t^{(2)}\\
        &- (\frac{p_1}{2p_1-1}\Delta s^{(1)} + \frac{p_2}{2p_2-1}\Delta s^{(2)})|\\
        &\leq \frac{1}{2p_1-1}\max\{|\Delta t^{(1)}+\Delta t^{(2)}|, |p_1\Delta s^{(1)} + p_2\Delta s^{(2)}|\}\\
        &\leq \frac{1}{2p_1-1}\tilde{d}_{max}.
    \end{aligned}
    \end{equation}
    Based on the Laplace mechanism, the existence of any edge $\left \langle v_i, v_j \right \rangle$ is protected by $(1-\alpha)\varepsilon_1$-LDP in Round 2.
    Following the combination theorem, we complete the proof of case 1.
    
    Case 2: $l(i,j) = 2$. Similarly, For any edge $\left \langle v_i, v_j \right \rangle$, it consumes privacy budget $\alpha\varepsilon_2$ in Round 1.
    Adding or removing edge $\left \langle v_i, v_j \right \rangle$ will lead to 
    $|\Delta s^{(2)}|\leq\tilde{d}_{max}$ and 
    $|\Delta t^{(2)}|\leq\tilde{d}_{max}$.
    We have:
    \begin{equation}
    \begin{aligned}
        |\Delta \widetilde{w}|
        &=|\Delta\widetilde{w}^{(2)}|\\
        &=|\frac{1}{2p_2-1}\Delta t^{(2)} -  \frac{p_2}{2p_2-1}\Delta s^{(2)}|\\
        &\leq \max\{|\frac{1}{2p_2-1}\Delta t^{(2)}|, |\frac{p_2}{2p_2-1}\Delta s^{(2)}|\}\\
        &\leq \frac{1}{2p_2-1}\tilde{d}_{max}.
    \end{aligned}
    \end{equation}
    According to Definition 3, Round 2 satisfies $(1-\alpha)\varepsilon_2$-LDP for any edge $\left \langle v_i, v_j \right \rangle$. Based on the combination theorem, we complete the proof of case 2.
\end{proof}

\setcounter{subfigure}{0}
\begin{figure}[t]
	\centering
	
	\subfigbottomskip=16pt %两行子图之间的行间距
	\subfigcapskip=-5pt %设置子图与子标题之间的距离
	\subfigure[Nodes $v_1$ at the end]{
	    \label{fig:uncertain1}
		\includegraphics[width=0.4\linewidth]{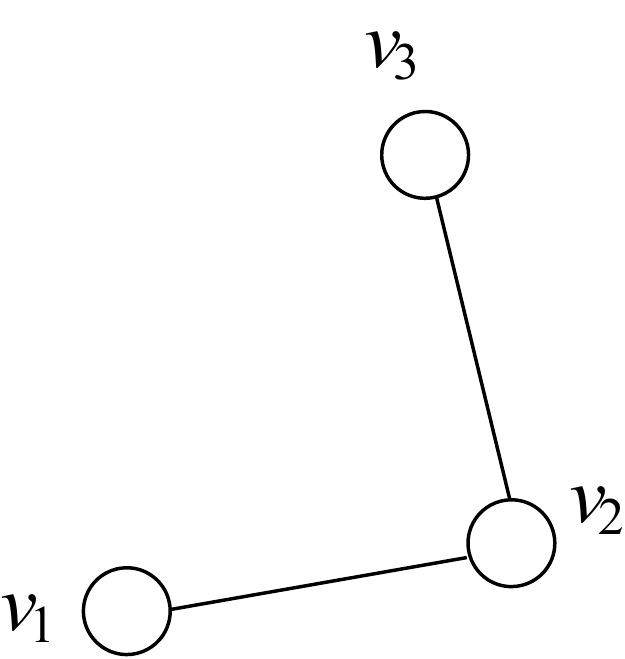}}
	\subfigure[Node $v_1$ in the middle]{
	    \label{fig:uncertain2}
		\includegraphics[width=0.4\linewidth]{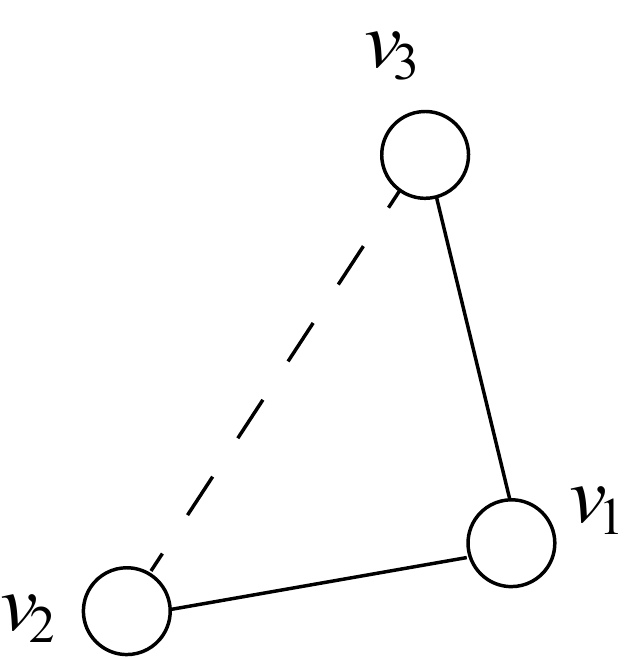}}

    \caption{Two graphs with the same network topology but different orders.}
\end{figure}

In decentralized social graphs, it is difficult to determine the exact variance of the proposed triangle counting algorithm. Even for the same network topology, the nodes' order can impact the estimation accuracy. For example, We consider two graphs with the same topology as shown in Fig.~\ref{fig:uncertain1} and \ref{fig:uncertain2}. In Fig.~\ref{fig:uncertain1}, node $v_1$ has only one adjacent edge, so it can't have any adjacent triangles. Moreover, $v_2$ and $v_3$ cannot form triangles with nodes with lower ordinals. Thus, there are no uncertain triangles in Fig.~\ref{fig:uncertain1}. In Fig.~\ref{fig:uncertain2}, for node $v_1$, there are two adjacent nodes whose ordinal number is larger than itself. Therefore, in local views of node $v_1$, there may be a triangle formed by $v_1,v_2$, and $v_3$ in the graph. It is evident that Fig.~\ref{fig:uncertain1} and \ref{fig:uncertain2} have the same network structure but different estimation errors. We give the error upper bound in Theorem \ref{thm:FGT-var}.

\begin{theorem}\label{thm:FGT-var}
    Let $n_1, n_2\geq0$ $(n_1+n_2=n)$ be the number of users with privacy protection at level 1, 2 respectively.
    For given $0<\varepsilon_1<\varepsilon_2$ (corresponding to the level), $\tilde{d}_{max}\geq d_{max}$ and $0<\alpha<1$, we have
    \begin{equation}
        \E[FGT(G,\varepsilon_1,\varepsilon_2)]=f_\bigtriangleup(G),
    \end{equation}
    and 
    \begin{equation}
        \V[FGT(G,\varepsilon_1,\varepsilon_2)]
        \leq O\left(n_1 \cdot f(\varepsilon_1) +n_2\cdot f(\varepsilon_2)\right),
    \end{equation}
    where $f(x) := \frac{e^{\alpha x}}{(e^{\alpha x}-1)^2}\left(\tilde{d}_{max}^3 +
        \frac{e^{\alpha x}}{((1-\alpha)x)^2}\tilde{d}_{max}^2\right)$ is a monotone decreasing function.
\end{theorem}
\begin{proof}
    First 
    \begin{equation}
    \begin{aligned}
    \E[FGT(G)]=&\E\bigg[\widetilde{w}^{(1)}+\widetilde{w}^{(2)}\\
        &+\sum_{i=1}^{n_1}{\rm Lap}\left(\frac{\tilde{d}_{max}/(2p_1-1)}{(1-\alpha)\varepsilon_1}\right)\\
        &+\sum_{i=n_1+1}^n {\rm Lap}\left(\frac{\tilde{d}_{max}/(2p_2-1)}{(1-\alpha)\varepsilon_2}\right) \bigg] \\
        &=\E[\widetilde{w}^{(1)}+\widetilde{w}^{(2)}].
    \end{aligned}
    \end{equation}
    Because $\E[s^{(1)}]=s^{(1)}$, $\E[s^{(2)}]=s^{(2)}$,
    \begin{equation}
        \E[t^{(1)}]=(1-p_1)(s^{(1)}-w^{(1)})+p_1w^{(1)},\\
    \end{equation}
    and
    \begin{equation}
        \E[t^{(2)}]=(1-p_2)(s^{(2)}-w^{(2)})+p_2w^{(2)},\\
    \end{equation}
    we have
    \begin{equation}
    \begin{aligned}
        &\E[\widetilde{w}^{(1)}+\widetilde{w}^{(2)}]\\
        &= \E[t^{(1)}-p_1s^{(1)}]+\E[t^{(2)}-p_2s^{(2)}]\\
        &=w^{(1)}+w^{(2)}\\
        &=f_\bigtriangleup(G).
    \end{aligned}
    \end{equation}
    
The variance is proved as follows:
    \begin{equation}
    \begin{aligned}
        &\V[FGT(G)]\\
        &=\V[\sum_{i=1}^n \hat{w}_i] \\
        &=\V[\sum_{i=1}^{n_1} \widetilde{w}_i+\sum_{i=n_1+1}^n \widetilde{w}_i]\\
        &+\V[\sum_{i=1}^{n_1}{\rm Lap}\left(\frac{\tilde{d}_{max}/(2p_1-1)}{(1-\alpha)\varepsilon_1}\right)\\
        &+\sum_{i=n_1+1}^n {\rm Lap}\left(\frac{\tilde{d}_{max}/(2p_2-1)}{(1-\alpha)\varepsilon_2}\right)]\\
        &=\V[\frac{1}{2p_1-1}(t^{(1)}-(1-p_1)s^{(1)})]\\
        &+\V[\frac{1}{2p_2-1}(t^{(2)}-(1-p_2)s^{(2)})]\\
        &+n_1\left(\frac{\tilde{d}_{max}/(2p_1-1)}{(1-\alpha)\varepsilon_1}\right)^2+n_2\left(\frac{\tilde{d}_{max}/(2p_2-1)}{(1-\alpha)\varepsilon_2}\right)^2\\
        &=\frac{1}{(2p_1-1)^2}\V[t^{(1)}]+\frac{1}{(2p_2-1)^2}\V[t^{(2)}]\\
        &+n_1\left(\frac{\tilde{d}_{max}e^{\alpha\varepsilon_1}}{(e^{\alpha\varepsilon_1}-1)(1-\alpha)\varepsilon_1}\right)^2
        +n_2\left(\frac{\tilde{d}_{max}e^{\alpha\varepsilon_2}}{(e^{\alpha\varepsilon_2}-1)(1-\alpha)\varepsilon_2}\right)^2
    \end{aligned}
    \end{equation}
Let 
\begin{equation*}
    c_{jk} = \sum_{i<j<k} \boldsymbol{1}(a_{i,j}=1\ \&\ a_{i,k}=1),
\end{equation*}
where $c_{jk}$ denotes the triangle counts that be affected by the edge $\left \langle v_j, v_k \right \rangle$, as shown in Fig.~\ref{fig: variance-root}.
Following the Bernoulli distribution, the variance of the value of $\left \langle v_j, v_k \right \rangle$ on the obfuscated matrix is $p_{l(j,k)}(1-p_{l(j,k)})$, $p_{l(j,k)}$ is the probability of retaining the origin value) regardless of whether it existed or not.

Therefore we get
\begin{equation}
\begin{aligned}
        &\frac{1}{(2p_1-1)^2}\V[t^{(1)}]+\frac{1}{(2p_2-1)^2}\V[t^{(2)}]\\
        &=\frac{1}{(2p_1-1)^2}p_1(1-p_1)\sum_{j<k, k\leq n_1}c_{jk}^2\\
        &+\frac{1}{(2p_2-1)^2}p_2(1-p_2)\sum_{j<k, k > n_1}c_{jk}^2\\
        &\leq\frac{p_1(1-p_1)\tilde{d}_{max}^3n_1}{(2p_1-1)^2}+\frac{p_2(1-p_2)\tilde{d}_{max}^3n_2}{(2p_2-1)^2}\\
        &\leq O\left(\frac{e^{\alpha\varepsilon_1}\tilde{d}_{max}^3}{(e^{\alpha\varepsilon_1}-1)^2}n_1+\frac{e^{\alpha\varepsilon_2}\tilde{d}_{max}^3}{(e^{\alpha\varepsilon_2}-1)^2}n_2\right)
\end{aligned}
\end{equation}
We complete the proof by combining the upper bound on the variance of the Laplace noise.

\end{proof}

\begin{figure}[!t]
	\centering
	\includegraphics[width=2.5in]{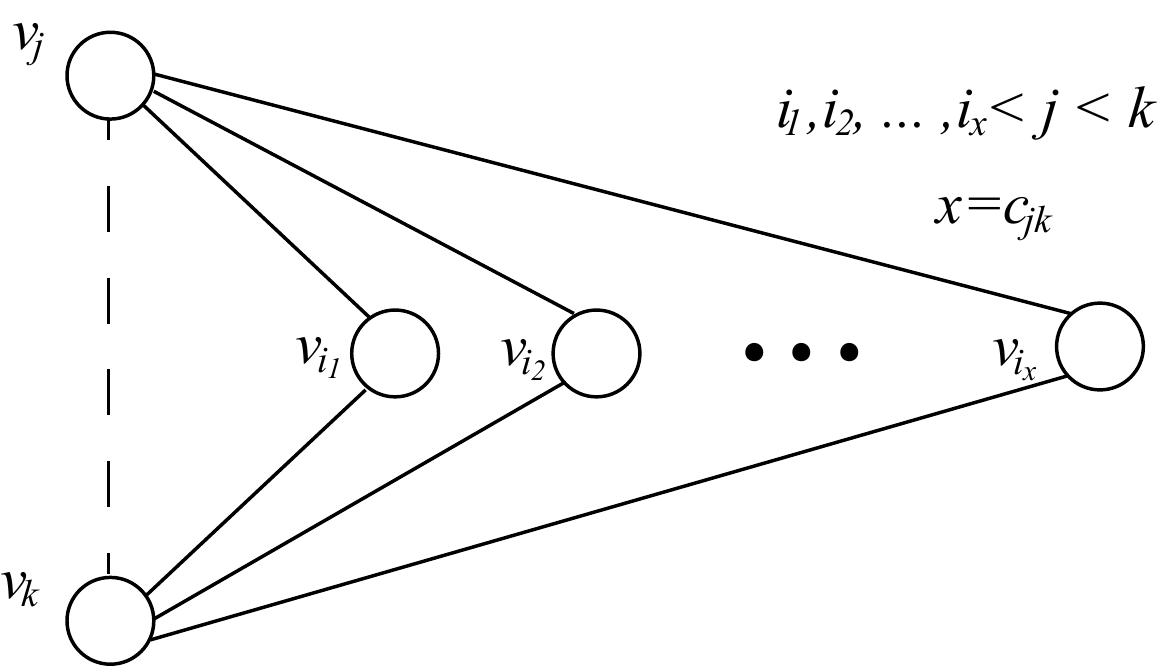}
	\caption{The affect of edge $\left \langle v_j, v_k \right \rangle$.}
	\label{fig: variance-root}
\end{figure}

\subsection{Muiti-level protection}
So far our discussion is limited to only two privacy levels in the system, i.e., $L=2$. We next show that our algorithms can be naturally extended to a multi-level privacy-preserving ($L\geq 3$) model for social graph data. We assume that there are $n$ nodes in the graph, and the number of nodes with the highest sensitivity level $i$ of adjacent edges is $n_i$, whose corresponding privacy budget is $\varepsilon_i$. Formally, we formulate $ n_i=|\{v_j|\min_{1\leq k\leq n,k\neq j}l(j,k)=i\}|,n=\sum_{i=1}^L n_i$. The smaller the privacy budget $\varepsilon_i$, the higher the privacy protection level. Our goal is to provide fine-grained privacy protection for different levels of edges. Then, we extend Algorithm 1 and Algorithm 2 to apply $k$-stars and triangle counting for multi-level privacy protection, respectively. We theoretically analyze their variances and compare them with existing methods.

\textbf{$k$-stars Counting:}
Since there are $L$ privacy levels, an extended version of Algorithm 1 would require that each node $v_i$ adds Laplacian noise $Lap(\frac{D}{\varepsilon_{l(v_i)}/2})$ to perturb the local view, where $l(v_i)\in (1,2,\cdots,L)$. The data analyst would aggregate the upload results in the same way as in Algorithm 1. We can easily derive the variance of the mechanism is
\begin{equation}
        \V[KS(G,\varepsilon_1,\varepsilon_2, ... ,\varepsilon_L)]
        = 4\dbinom{\tilde{d}_{max}}{k-1}^2\sum_{i=1}^L(\frac{n_i}{\varepsilon_i^2}).
    \end{equation}

\textbf{Triangle Counting:}
We extend Algorithm 2 to multi-level triangle counting. To achieve fine-grained privacy protection, the extended version of Algorithm 2 would require each node to perform the following:
\begin{enumerate}
    \item Initialize $p_l=\frac{e^{\alpha \varepsilon_l}}{e^{\alpha \varepsilon_l}+1}$, $l\in (1,2,\cdots,L)$. The data analyst rearranges each node based on the privacy level, and then each node obfuscates the local view based on its serial number and privacy level before uploading.
    \item In the second round, the data analyst divides the nodes into different zones $\{(1,n_1),(n_1+1,n_1+n_2),\cdots,(\sum_{l=1}^L n_l+1,n) \}$ according to the privacy level. The result $\hat{w}_i$ uploaded by the node in the $i$-th interval (corresponding to the privacy budget is $\varepsilon_i$) is 
    \begin{equation}
    \begin{aligned}
      \hat{w}_i = &\widetilde{w}_i^{(l(v_i))} + \widetilde{w}_i^{(l(v_i)+1)}+ \cdots + \widetilde{w}_i^{(L)}\\ & + {\rm Lap}(\frac{\tilde{d}_{max}/(2p_{l(v_i)}-1)}{(1-\alpha)\varepsilon_{l(v_i)}})
    \end{aligned}
    \end{equation} 
\end{enumerate}

Once the data analyst collects data from all nodes, she can estimate the triangle counts $\sum_{i=1}^n \hat{w}_i$ in the same way as in Algorithm 2. The accuracy of this estimate is determined by the number of edges with high sensitivity levels, and such edges represent only a small fraction in decentralized social graphs. Thus, our scheme can substantially reduce the error in graph analysis. Obviously, the variance of the mechanism is at most
\begin{equation}
        \V[FGT(G,\varepsilon_1,\varepsilon_2,...,\varepsilon_L)]
        \leq O\left(\sum_{i=1}^L n_i \cdot f(\varepsilon_i) \right),
\end{equation}
    where $f(x) := \frac{e^{\alpha x}}{(e^{\alpha x}-1)^2}\left(\tilde{d}_{max}^3 + \frac{e^{\alpha x}}{((1-\alpha)x)^2}\tilde{d}_{max}^2\right)$ is a monotone decreasing function.

\begin{table}[ht]\label{l2loss}
\caption{Upper Bounds on MSE for privately estimating $f_{k\hollowstar}$ and $f_\bigtriangleup$}
\center
\renewcommand\arraystretch{1.5}
\setlength{\tabcolsep}{7mm}{
\begin{tabular}{|c|c|c|c|}\hline
  & Local2Rounds & Our Scheme   \\ \hline
$f_{k\hollowstar}$   & $O(n/\varepsilon^2)$     & $O(\sum_{i=1}^L(n_i/\varepsilon_i^2))$   \\ \hline
$f_\bigtriangleup$   & $O\left(n \cdot f(\varepsilon)\right)$     & $O\left(\sum_{i=1}^L n_i f(\varepsilon_i) \right)$    \\ \hline
\end{tabular}
}
\end{table}

We compare the upper bounds on the MSE of our scheme with Local2Rounds for privately estimating $f_{k\hollowstar}$ and $f_\bigtriangleup$ in Table II, where $n = \sum_{i=1}^{L}n_i, \varepsilon_i \leq \varepsilon$. Clearly, our scheme has a lower upper bound than Local2Rounds’ due to the fact that we provide fine-grained privacy protection for edges with different privacy levels, while Local2Rounds would overprotect the edges with low sensitivity.

\section{Experiments}\label{Experiments}
\subsection{Experiment settings}

\textbf{Datasets.} We perform experiments on two real-world datasets from {\itshape Snap datasets: Stanford large network dataset collection} \cite{datasets}. The LiveJournal database (LJDB) includes “friends list” from LiveJournal, a free online blogging community where users proclaim friendships with each other. The Orkut database derives from a free online social network where users can make friendships. Orkut also allows users to create groups that other members can join. Table III shows the properties of the datasets. 

\begin{table}[ht]\label{datasets}
\caption{Dataset Properties}
\center
\renewcommand\arraystretch{1.5}
\setlength{\tabcolsep}{2.7mm}{
\begin{tabular}{|c|c|c|c|}\hline
Dataset  & Number of nodes   &   Number of edges  &   Average degree     \\ \hline
LJDB     &   3997962       &       34681189   &    8.67            \\ \hline
Orkut    &   3072441       &      117185083   &    38.1             \\ \hline
\end{tabular}
}
\end{table}

Obviously, LJDB is more sparse than Orkut. For each database, we randomly select $n$ users from the global graph and extract a graph $G=(V,E)$ with $n$ users. We then estimate the number of triangles $f_\bigtriangleup$ and $k$-stars $f_\hollowstar$ using the algorithm in Section V.

\textbf{Parameter selection.}  Notice that our notion FGR-DP is generally suitable for multi-level privacy protection in social graph analysis. Therefore, we have to allocate multiple privacy budgets for different sensitivity levels. In the experiment, we assume there are two privacy levels with privacy budget$ \{\varepsilon_1, \varepsilon_2\}$ (as we describe in Section V), and we set  $\varepsilon_2= 2\varepsilon_1$. The privacy budget for any edges is randomly selected from the two values with a specific budget distribution, where the default distribution is $\{20\%, 80\% \}$, and we will vary the budget distribution in our experiments to evaluate the impact.  Intuitively, we consider that 20$\%$ of the edges in the social graph have higher privacy requirements than the remaining edges.  

Except for different privacy levels, since our algorithm contains two rounds of interactions in triangle counting, we need to split the privacy budget among different rounds. Recall from section V that Round 1 estimates a confused adjacency matrix in the two-round interaction, and Round 2 reports noise counts. As we value Round 1 and Round 2 equally in triangle counting, $\alpha$ is set to 0.5. Finally, for ease of computation, we use the maximum degree $d_{max}$ as the global sensitivity, i.e., $\tilde{d}_{max}=d_{max}$.

\begin{figure}[!t]
  \centering
  \subfigure[LJDB(triangle, $n=10000$)]{\includegraphics[width=4cm]{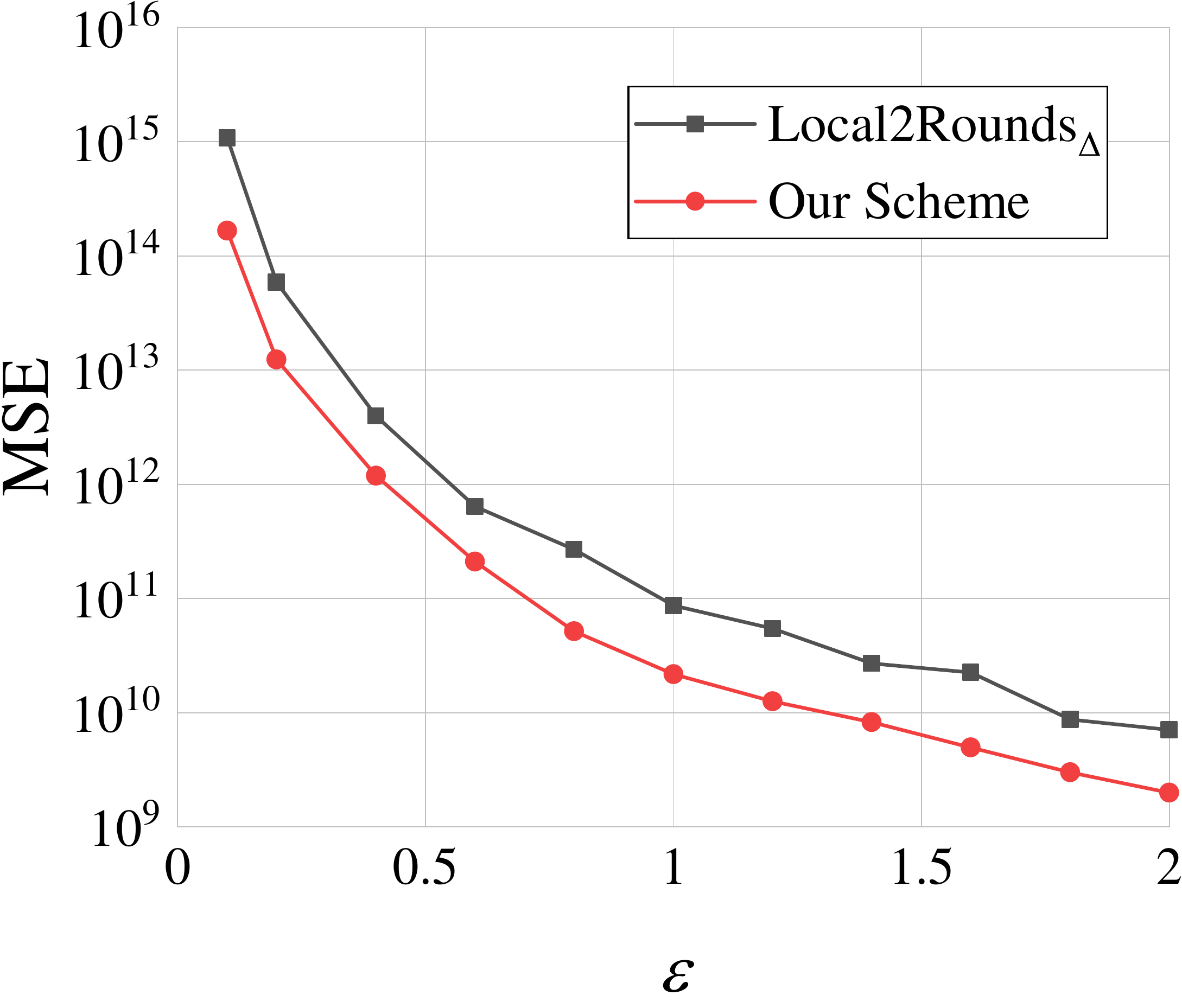}} 
  \subfigure[Orkut(triangle, $n=10000$)]{\includegraphics[width=4cm]{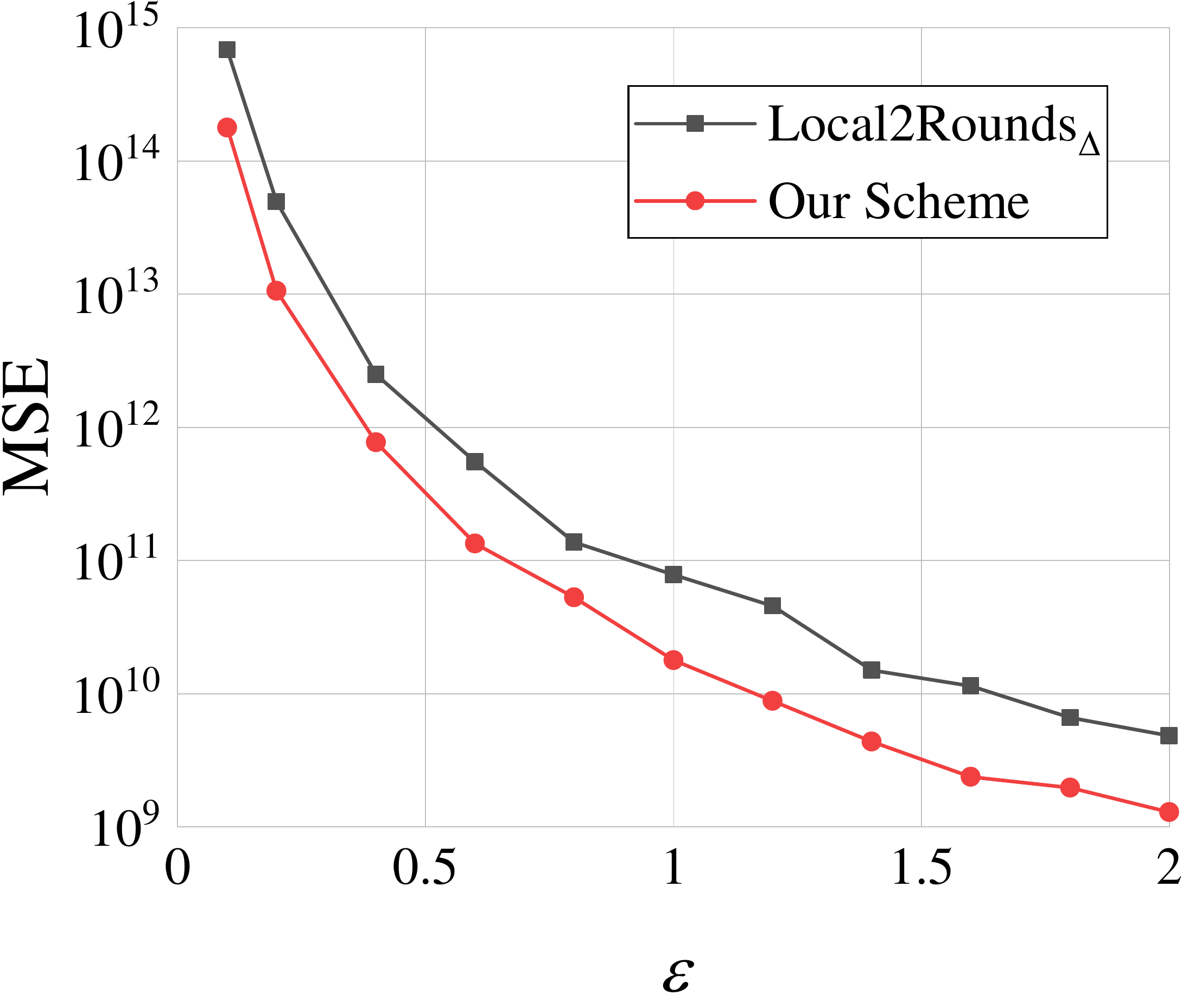}} \\
  \caption{Relation between $\varepsilon$ and the MSE in triangle counting when $n = 10000, \alpha=0.5$)}
  %\label{fig:data_distribution}
  \vspace{-0.2in}
\end{figure} 

\begin{figure}[!t]
  \centering
  \subfigure[LJDB(2-star, $n=24000$)]{\includegraphics[width=4cm]{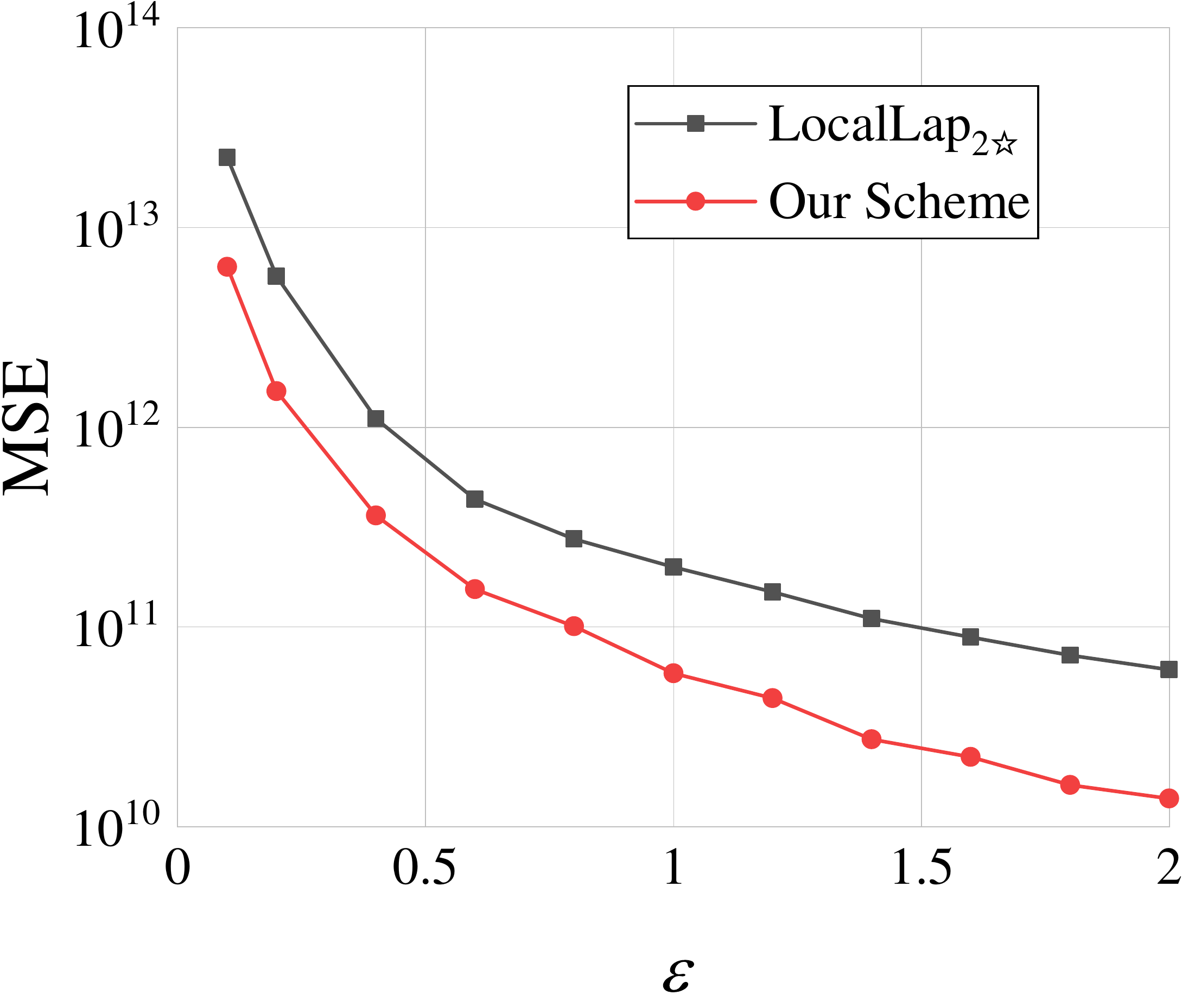}} 
  \subfigure[Orkut(2-star, $n=24000$)]{\includegraphics[width=4cm]{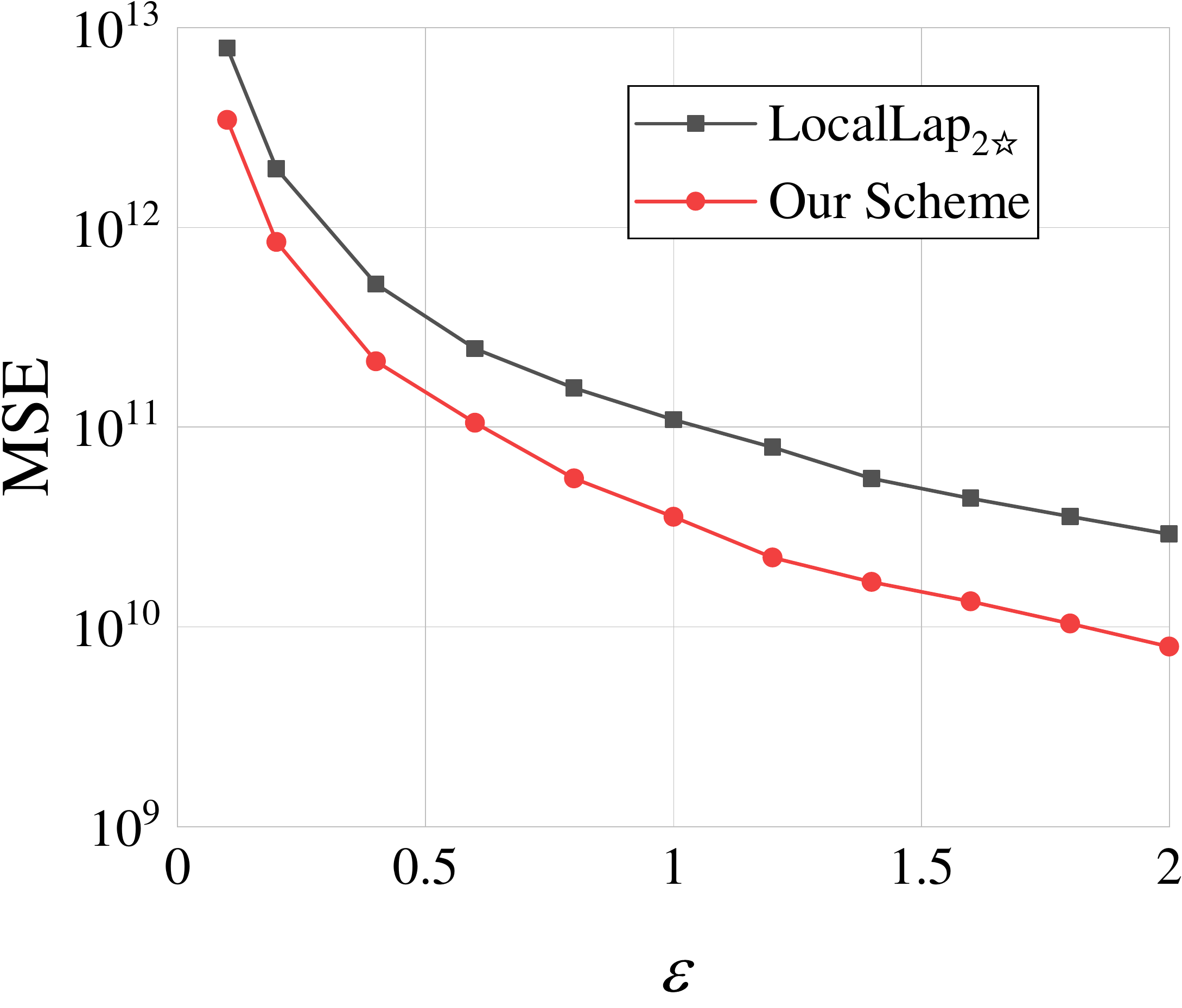}} 
  \caption{Relation between $\varepsilon$ and the MSE in 2-star counting when $n = 24000$}
 % \label{fig:data_distribution}
  %\vspace{0.2in}
\end{figure} 

\textbf{Utility Metrics.}
 Let $\hat{f}(G) \in \mathbb{R}$ be an estimate of global subgraph count $f(G) \in \mathbb{R}$, where $f$ can be instantiated by $f_\bigtriangleup$ or $f_\hollowstar$. We employ the mean squared error (MSE) and the mean relative error (MRE) as utility metrics to evaluate the accracy of our estimation, defined by
 \begin{equation}
     {\rm MSE} = \dfrac{(\hat{f}(G)-f(G))^2}{f(G)}, \quad {\rm MRE} = \dfrac{|\hat{f}(G)-f(G)|}{f(G)}
 \end{equation}
  where $f(G) \neq 0$. The smaller the MRE(MSE), the more accurate the estimated results are. All experimental results are averaged with 100 repeats.

\subsection{Experimental Results}
\textbf{Relation between $\varepsilon$ and the MSE.} We first evaluate the MSE of the estimates of $f_\bigtriangleup$, $f_{2\hollowstar}$ when we change the privacy budget $\varepsilon$. We omit the result of 3-stars because it is similar to that of 2-stars. Fig. 3 depicts MSE values output of triangle counts by our scheme and baseline described above when the privacy budget varies from 0.1 to 2. We set $\alpha=0.5$ because we value the adjacency matrix collection and noise counts collection equally in the triangle counting algorithm. The results show that our scheme achieves higher accuracy over all datasets. Note that the difference is clear since MSE is plotted in log-scale. By contrast, the proposed algorithm significantly outperforms the baseline method in terms of overall accuracy (the MSE is usually only one-fifth of the baseline). Similarly, we compare the MSE of our 2-stars collection algorithm with $LocalLap_{2 \hollowstar} $ over all datasets in Fig. 4. Fig.3 and 4 are roughly consistent with the upper bounds. Notice that all experiments have smaller MSE in Orkut than LJDB because the latter is more sparse than the former. 

\begin{figure}[!tb]
  \centering
    \subfigure[LJDB(triangle)]{\includegraphics[width=4cm]{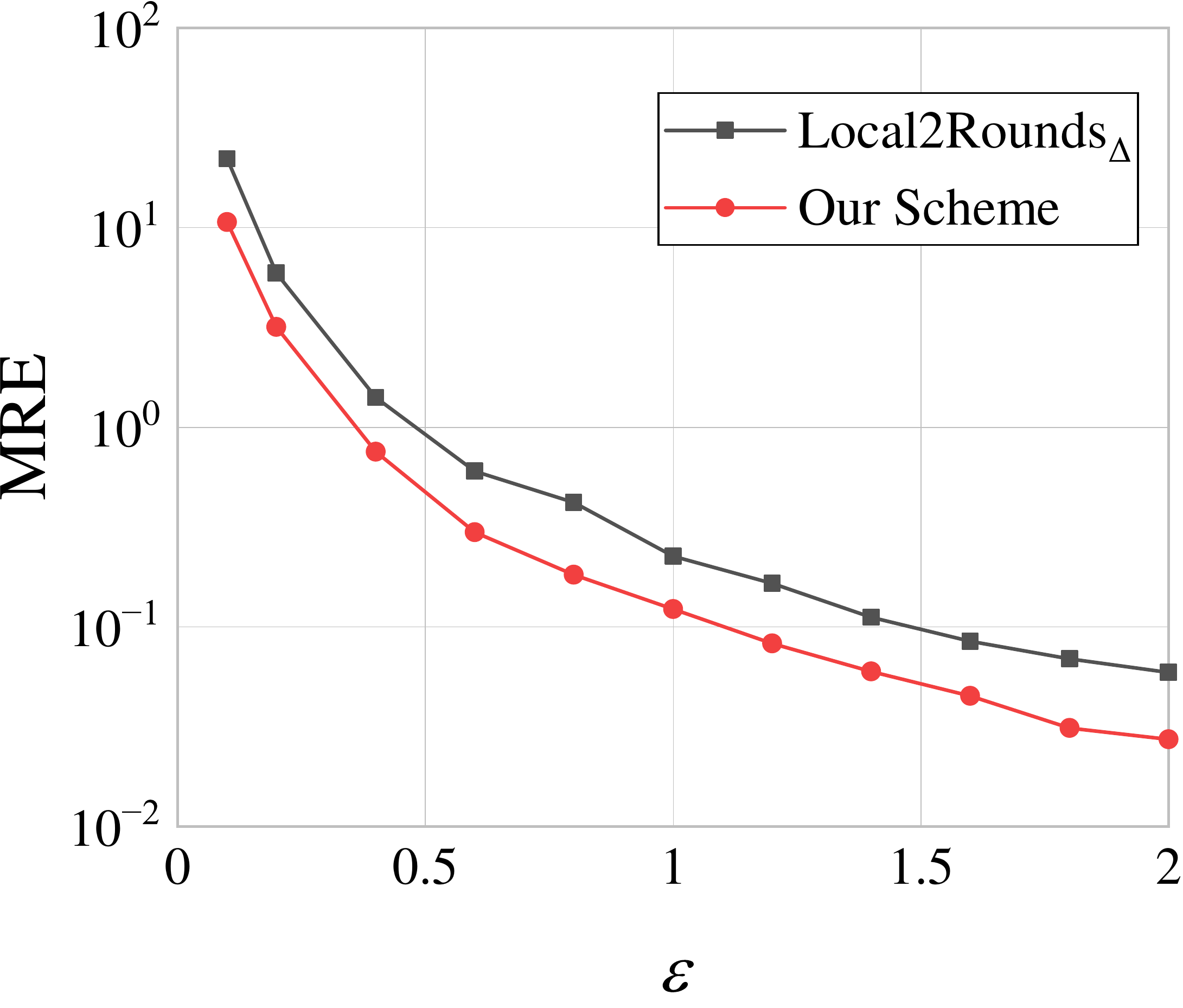}} 
    \subfigure[Orkut(triangle)]{\includegraphics[width=4cm]{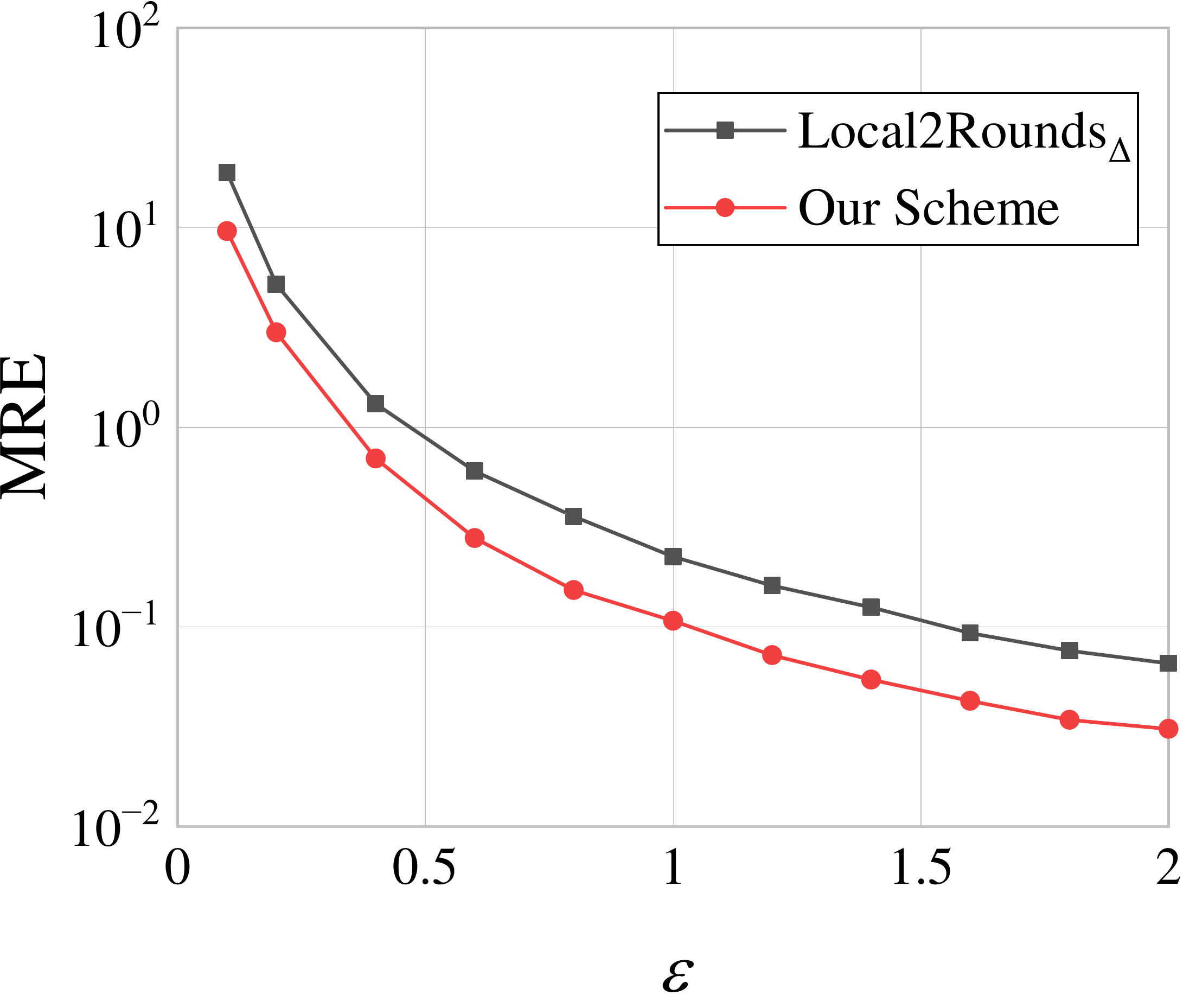}}\\
    \subfigure[LJDB(2-star, n=24000)]{\includegraphics[width=4cm]{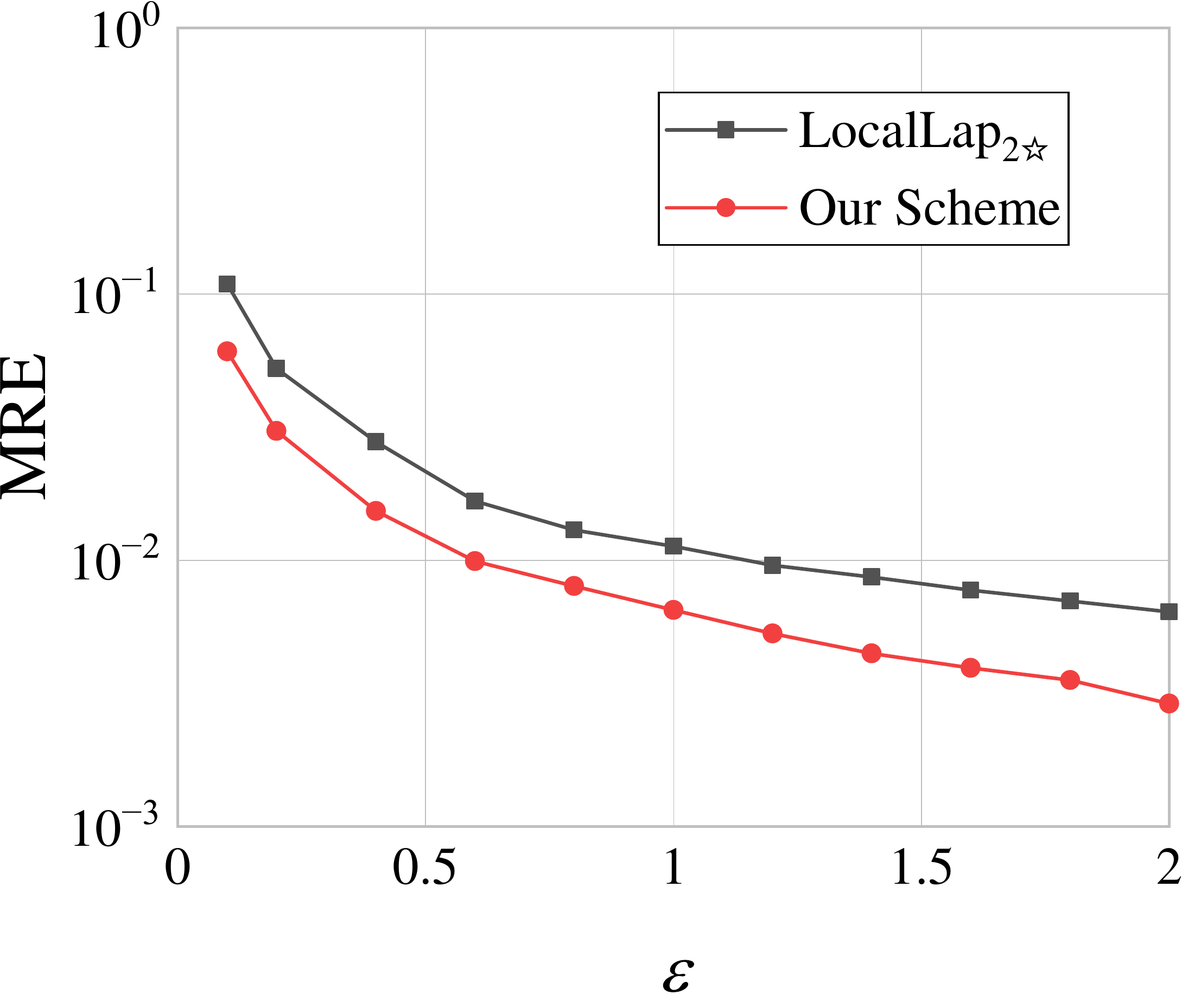}} 
    \subfigure[Orkut(2-star, n=24000)]{\includegraphics[width=4cm]{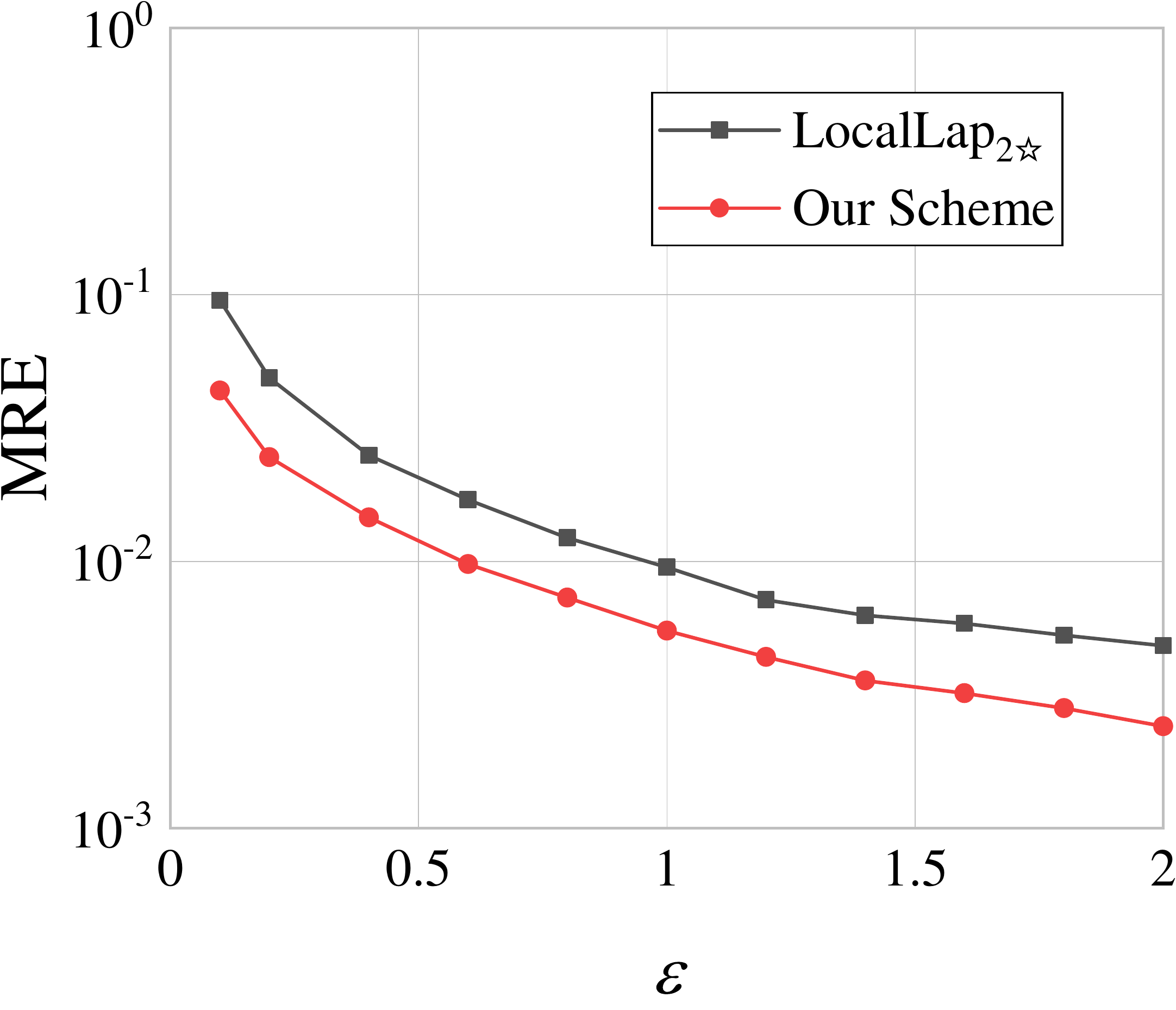}}\\
  \caption{Relation between $\varepsilon$ and the MRE.}
 % \label{fig:data_distribution}
  \vspace{-0.2in}
\end{figure} 

\textbf{Relation between $\varepsilon$ and the MRE.} 
 When the number of subgraphs in the social graph is large, the MSE will also be enormous. Thus, we also employ the mean relative error (MSE) as our utility metrics, as described above. Fig. 5 shows the relation between $\varepsilon$ and the MRE for triangle counting and 2-stars counting. Again, the figure shows that our algorithms achieve better accuracy over all datasets. Similarly, we plot MRE in log-scale, decreasing gradually as the privacy budget increases. Obviously, the MRE always remains below 10\% in triangle counting regardless of datasets when the privacy budget is relatively large, e.g., $\varepsilon=1$, as illustrated in Fig. 5(a) and 5(b). This trend is more evident in the 2-stars counting algorithm. We can directly observe from Fig. 5(c) and 5(d) that when the privacy budget $\varepsilon$ = 1, its MRE is always below or close to 0.55\% over all datasets. When $\varepsilon$ decreases, the accuracy reduces, but the MRE is still lower than 4.9\% even when $\varepsilon=0.1$.

\begin{figure}[!tb]
  \centering
    \subfigure[LJDB($\varepsilon_1=0.5, \varepsilon_2=1$)]{\includegraphics[width=4cm]{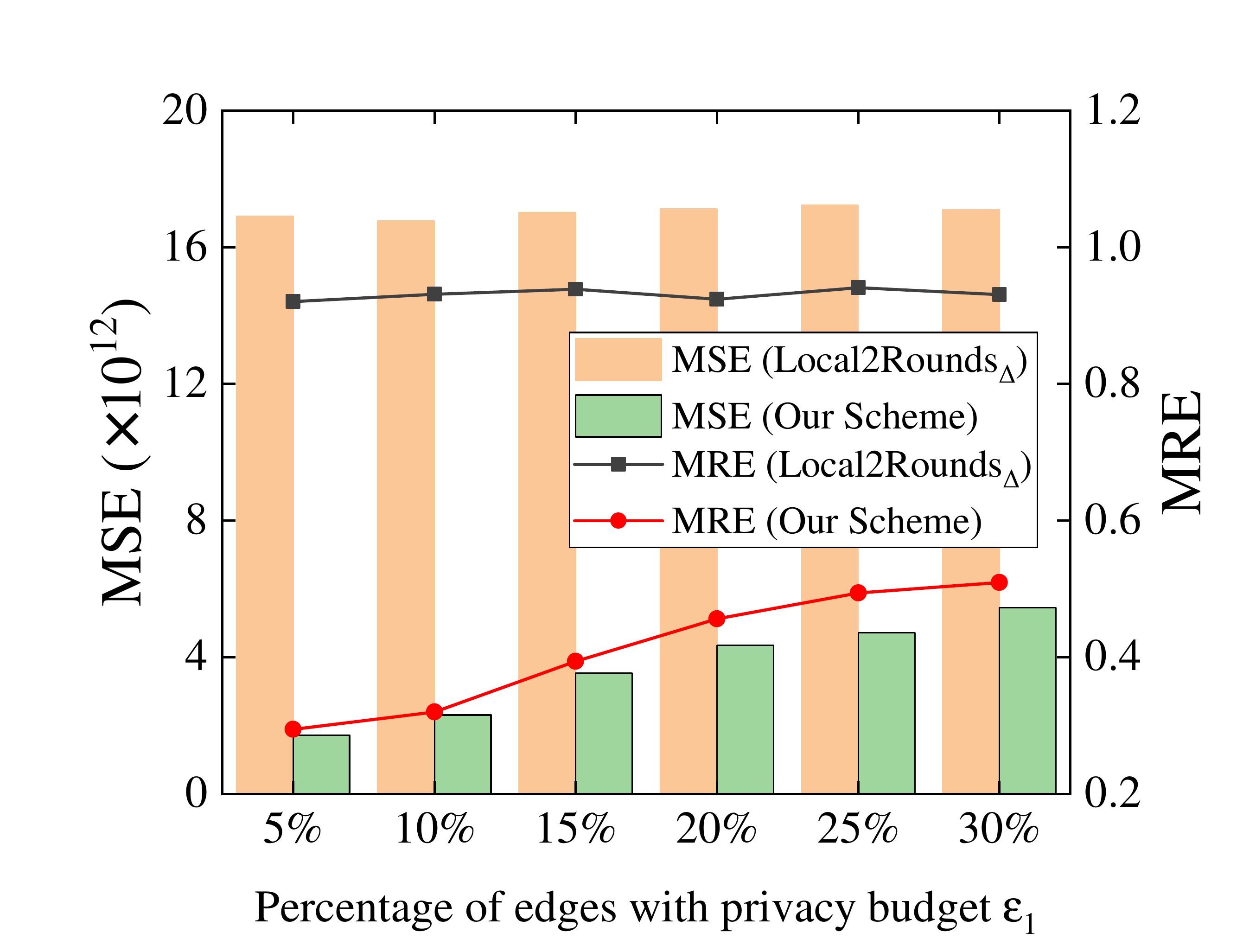}} 
    \subfigure[Orkut($\varepsilon_1=0.5, \varepsilon_2=1$)]{\includegraphics[width=4cm]{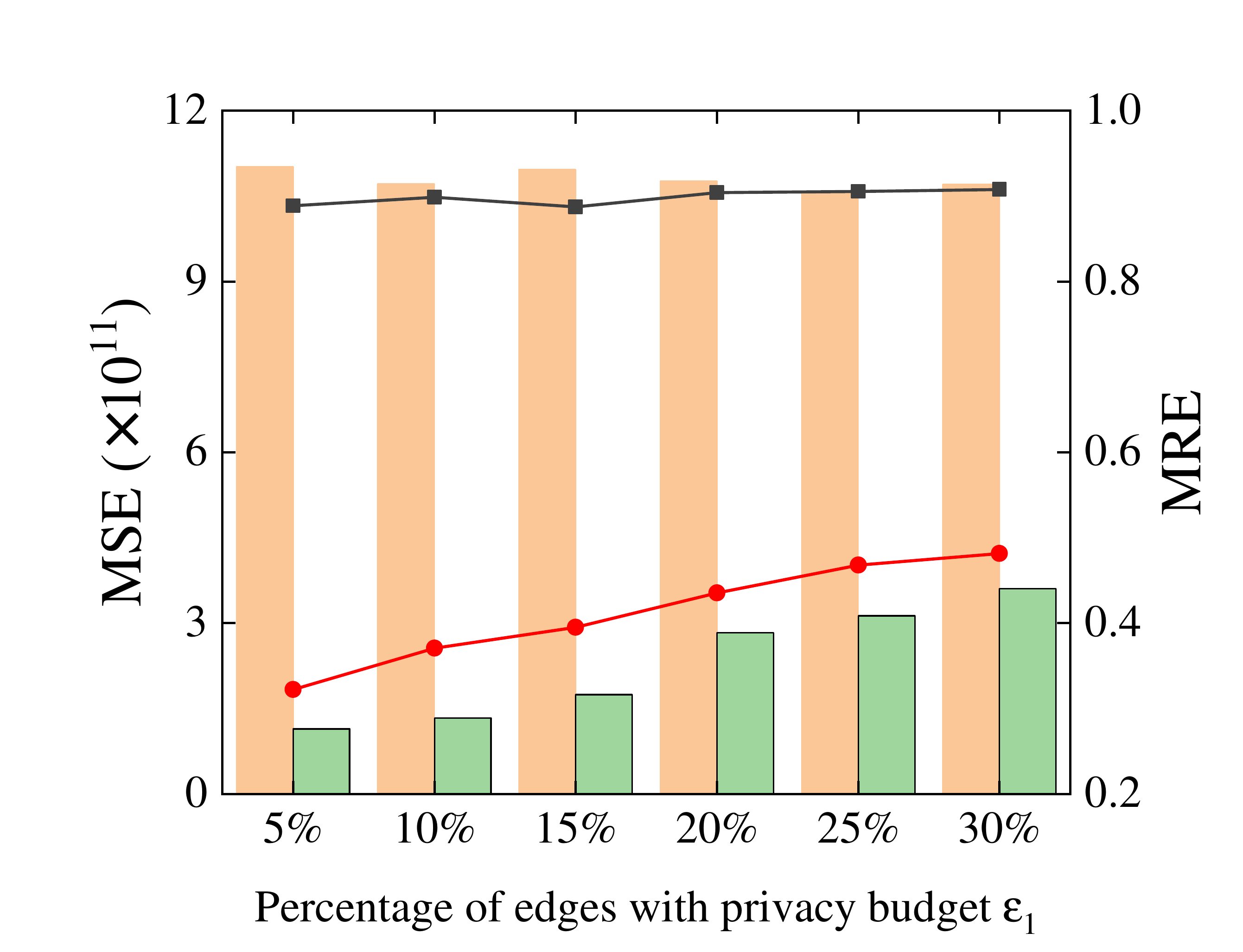}}\\
    \subfigure[LJDB($\varepsilon_1=1, \varepsilon_2=2$)]{\includegraphics[width=4cm]{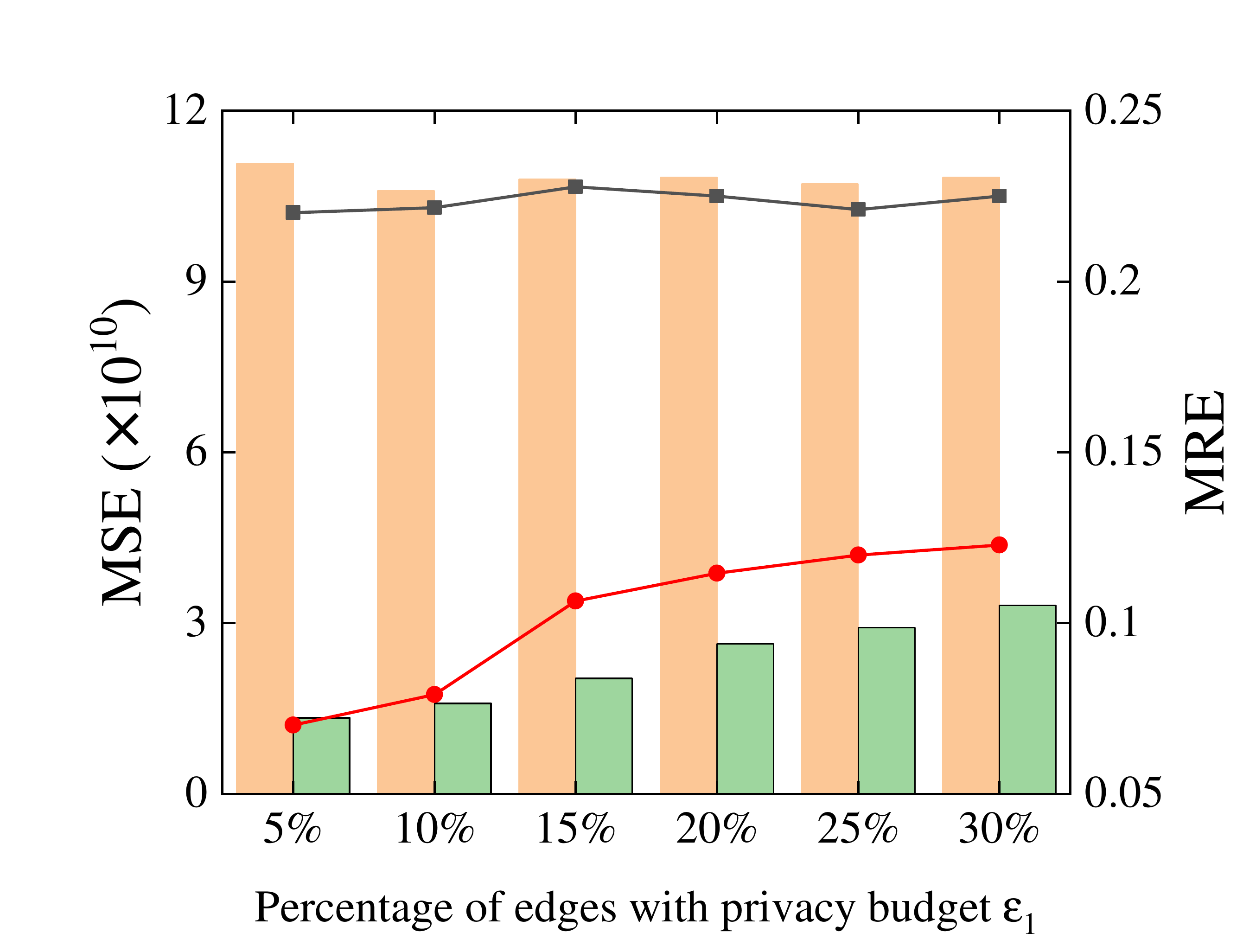}} 
    \subfigure[Orkut ($\varepsilon_1=1, \varepsilon_2=2$)]{\includegraphics[width=4cm]{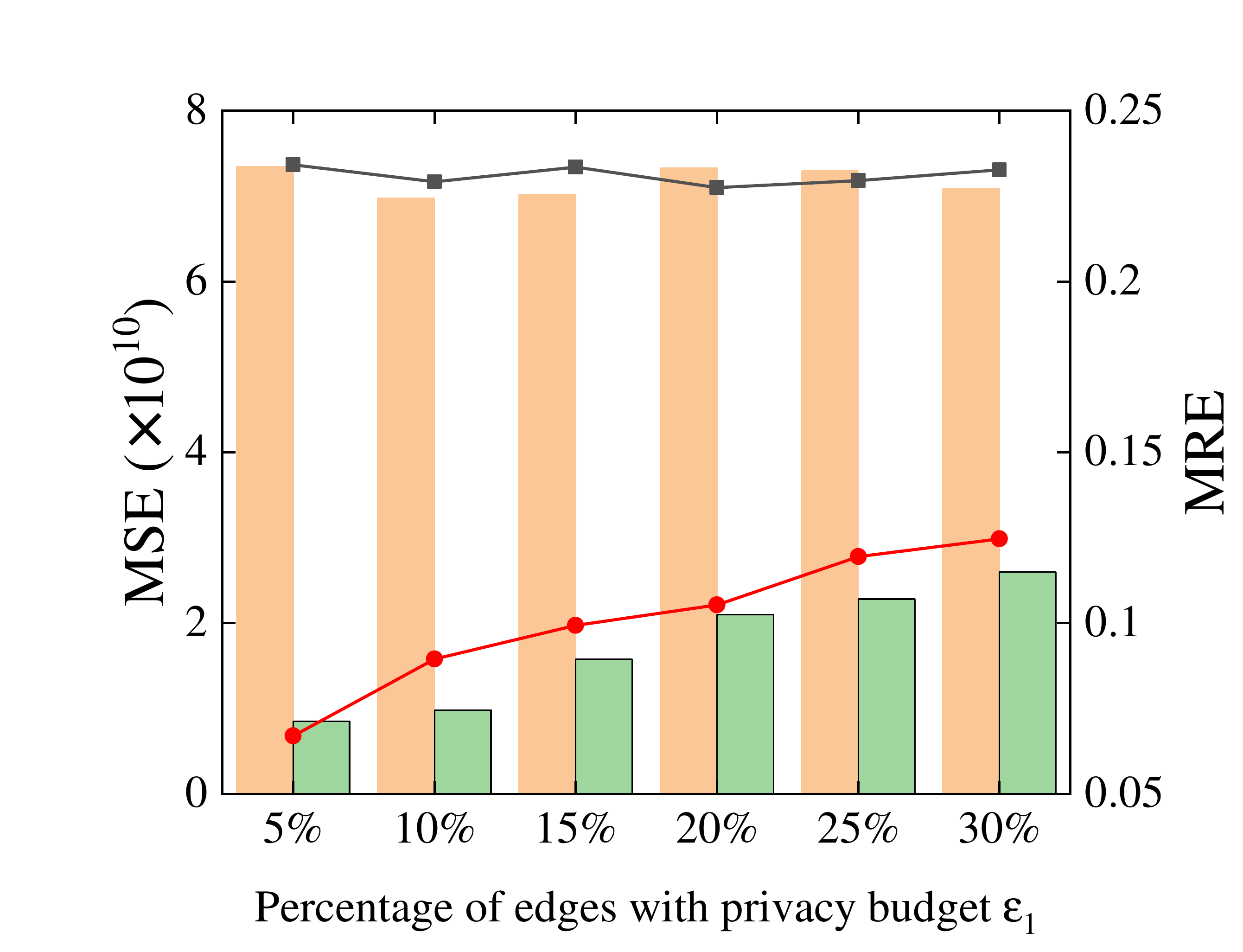}}\\
  \caption{Under different privacy budget distributions, when $n=10000$.}
  %\label{fig:data_distribution}
 % \vspace{in}
\end{figure} 

\begin{figure}[!tb]
  \centering
    \subfigure[LJDB(triangle)]{\includegraphics[width=4cm]{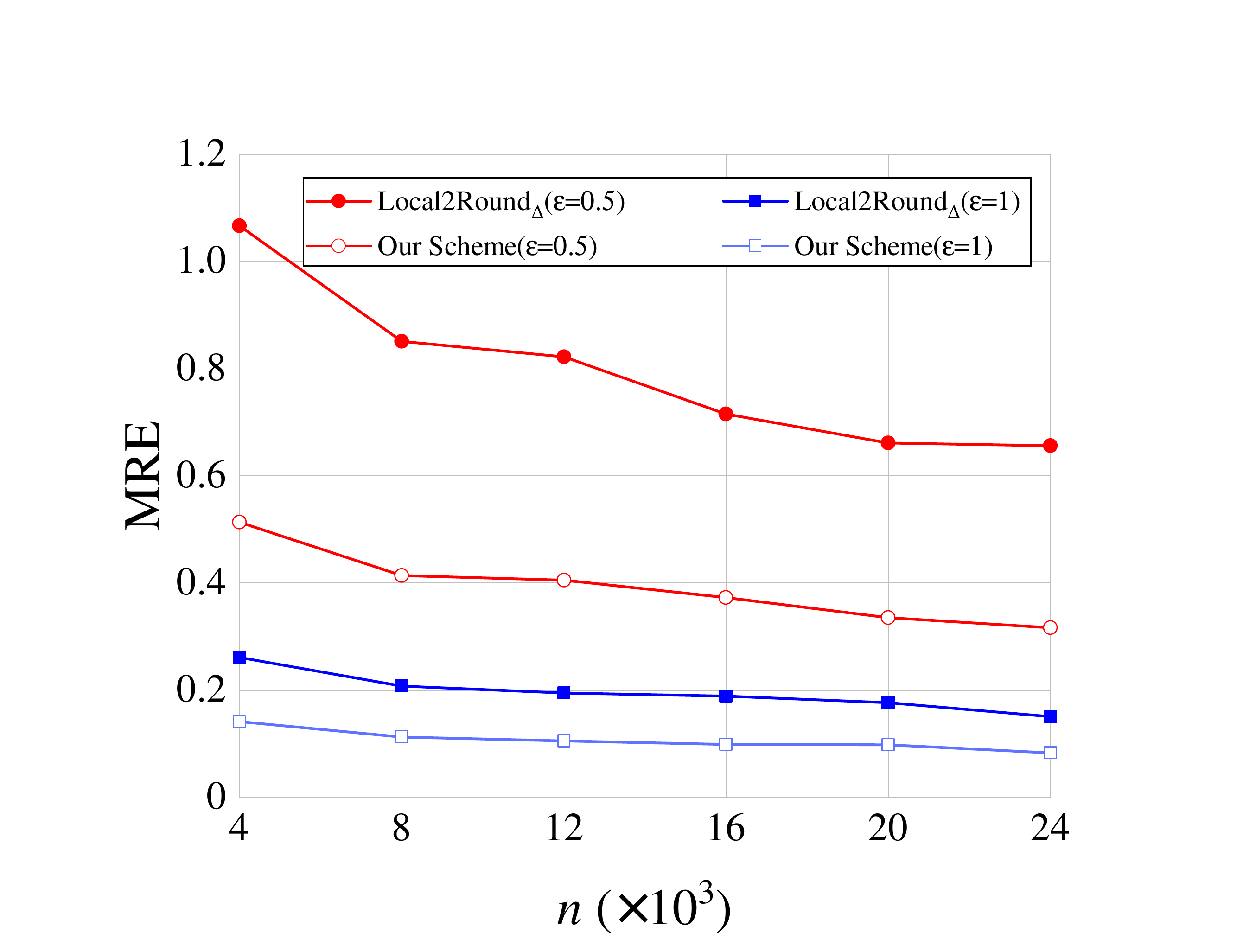}} 
  \subfigure[Orkut(triangle)]{\includegraphics[width=4cm]{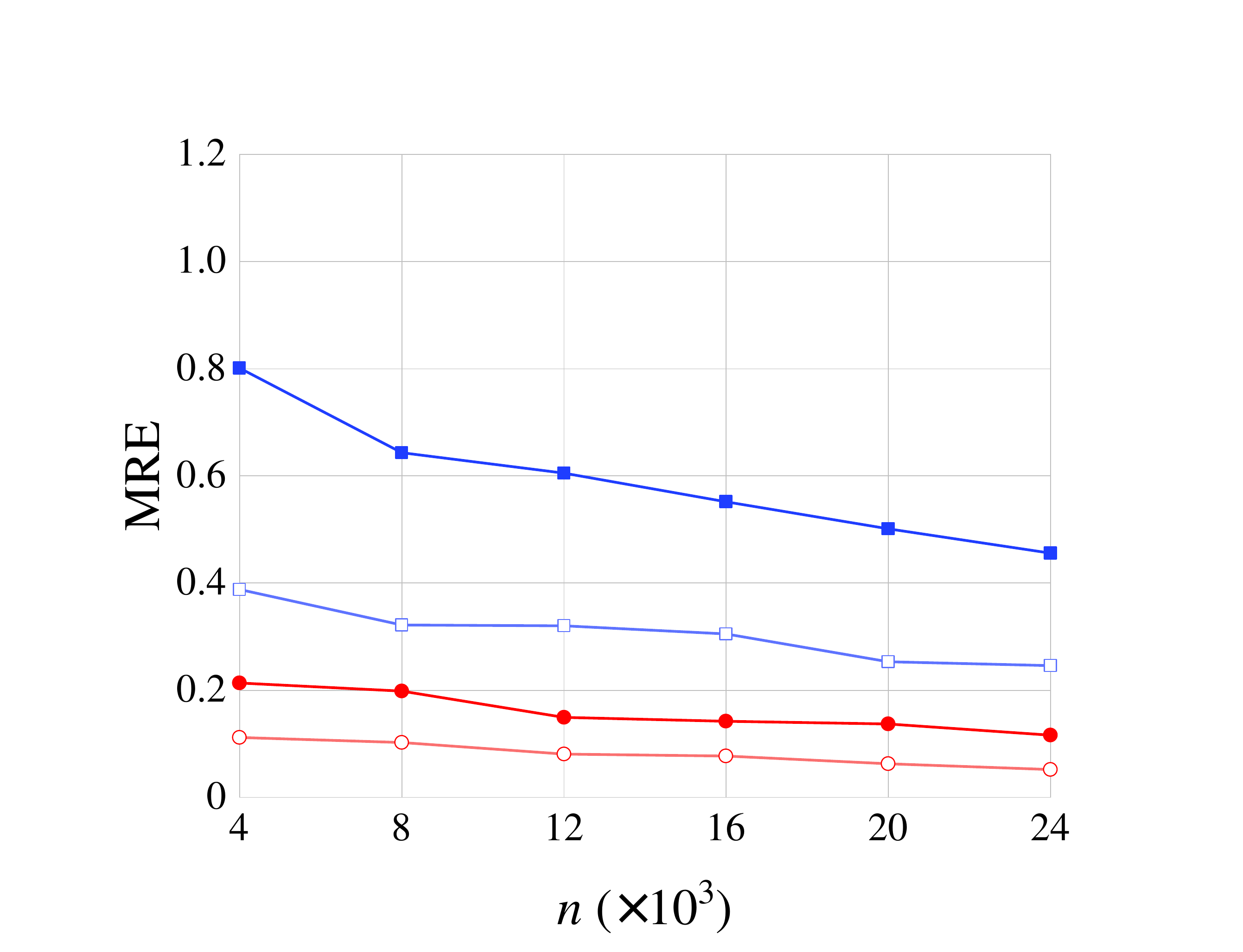}} 
  \caption{triangle Relation between $n$ and the relative error.}
 % \label{fig:data_distribution}
  %\vspace{0.2in}
\end{figure} 

\textbf{Influence of Privacy Budget Distributions.} 
Fig. 6 shows the MSE and MRE (with $n$=10000) under different privacy budget distributions in the triangle counting algorithm over all datasets. The $k$-stars result is not included in this case because it is similar to the triangle result. We change the percentage of edges whose privacy budget is $\varepsilon_1$ (the smaller one) from 5\% to 30\% with a 5\% increment. 
This setting is reasonable because the edges with high sensitivity levels usually account for a low percentage in the social graph. Under a relatively higher percentage, i.e., many edges are more sensitive than others (e.g., 20\%), our algorithm can also greatly improve the accuracy over all datasets. The error of our method gradually decreases as the percentage decreases. However, the baseline method always maintains a higher error constant.
Furthermore, we change $\varepsilon_1$ to observe the effect of the privacy budget on the results. Obviously, we can find that they share a similar trend in the Orkut dataset by comparing Fig. 6(b) and 6(d). Therefore, our algorithm achieves good accuracy regardless of the privacy budget.

\textbf{Relation between $n$ and the MRE.}
Fig. 7 and 8 describe the relation between $n$ and the MRE in triangle counting and $k$-stars counting, respectively. We can observe that the MRE decreases as $n$ increases for all cases because when $n$ increases, both $f_\bigtriangleup(G)$ and $f_\hollowstar(G)$ increase significantly. Another observation is that the MRE in Orkut is smaller since Orkut is denser and contains more triangles and $k$-stars; i.e., the denominator of the MRE is very large. What's more, for $k$-stars counting in Fig. 8, the MRE of our method at $\varepsilon=0.5$ is very close to that of $LocalLap_{2\hollowstar}$ at $\varepsilon=1$ in the case of the LJDB dataset, which fully demonstrates that the proposed algorithm can achieve higher accuracy with a smaller privacy budget.

\begin{figure}[!tb]
  \centering
    \subfigure[LJDB(2-star)]{\includegraphics[width=4cm]{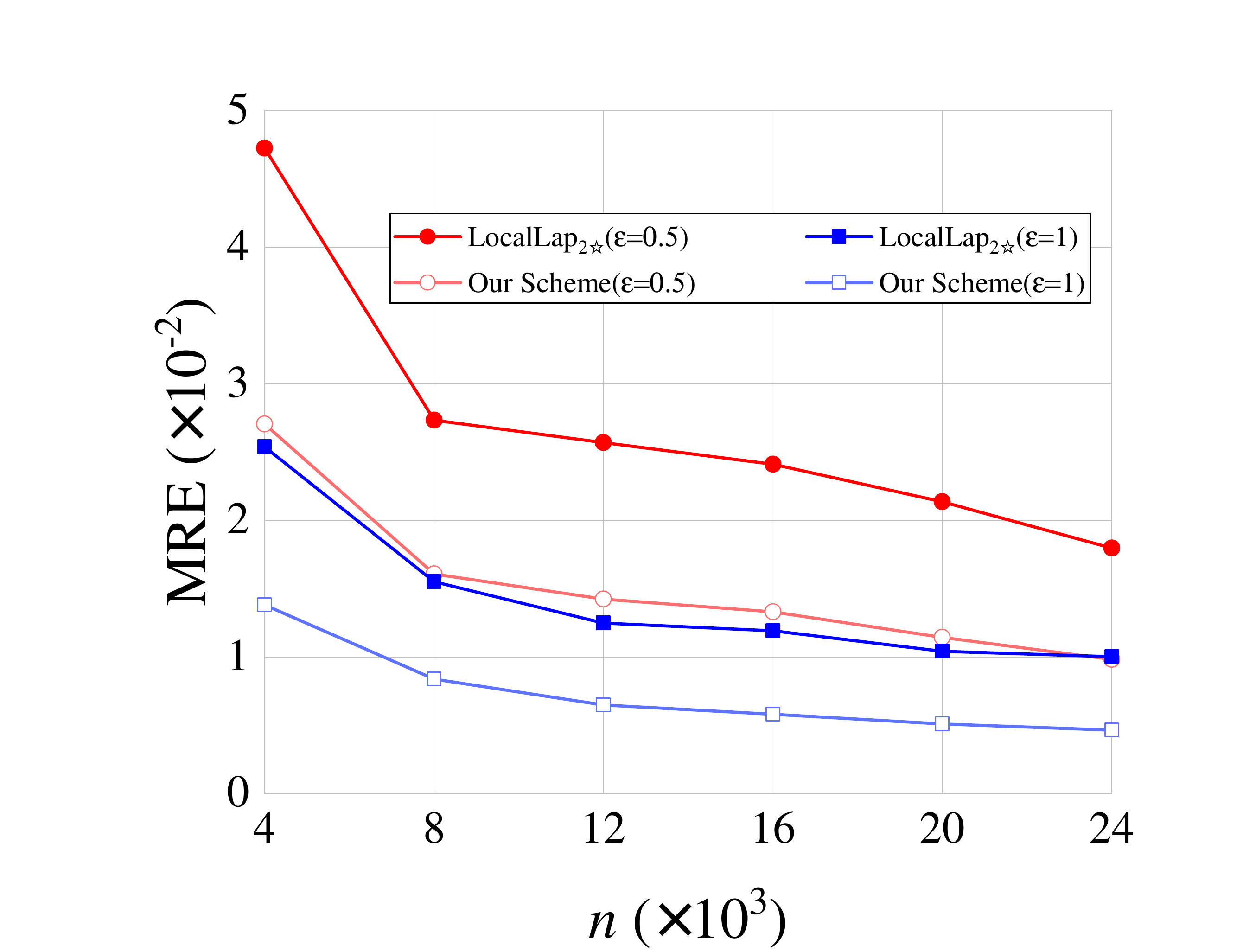}} 
    \subfigure[Orkut(2-star)]{\includegraphics[width=4.13cm]{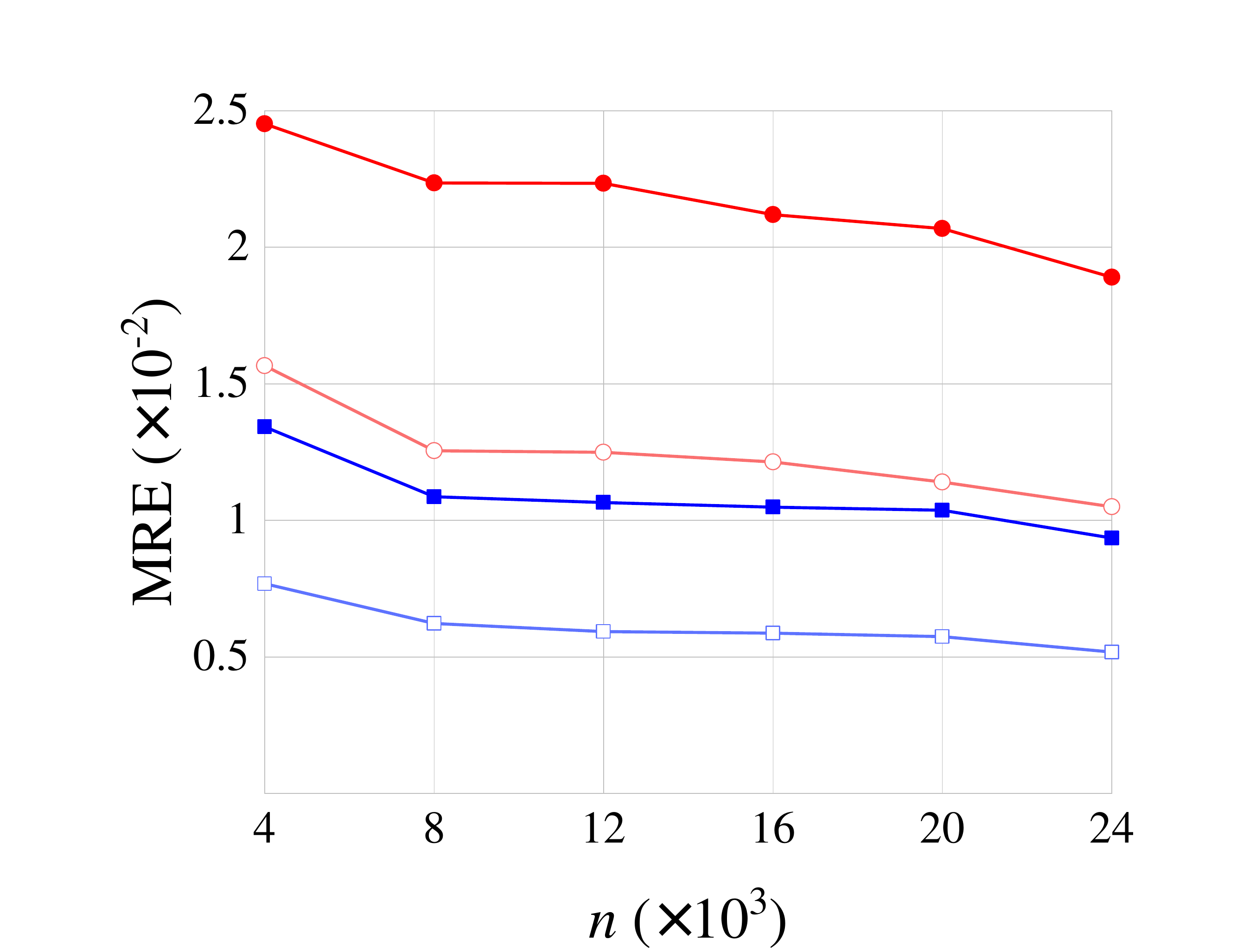}} \\
    \subfigure[LJDB(3-star)]{\includegraphics[width=4cm]{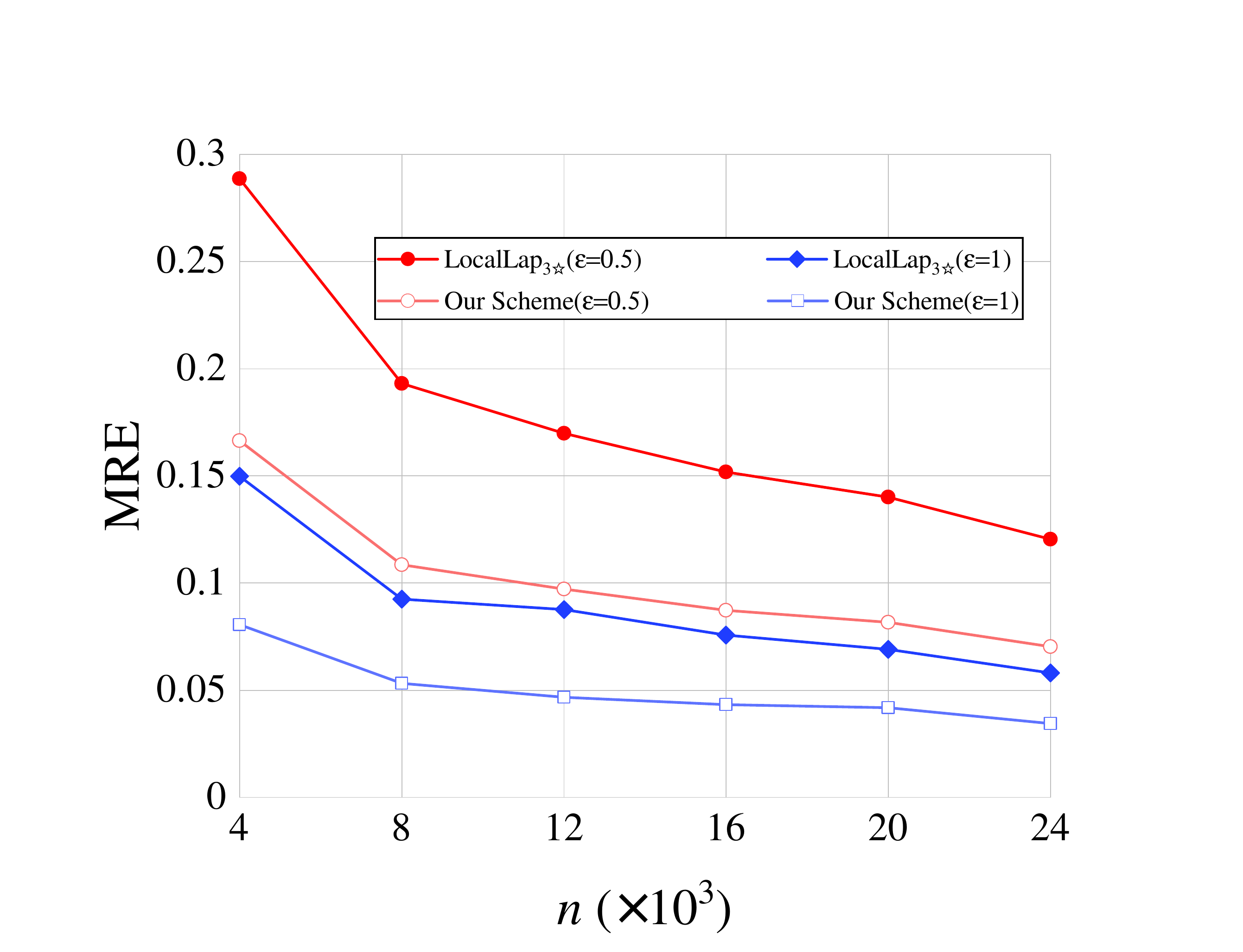}} 
    \subfigure[Orkut(3-star)]{\includegraphics[width=4cm]{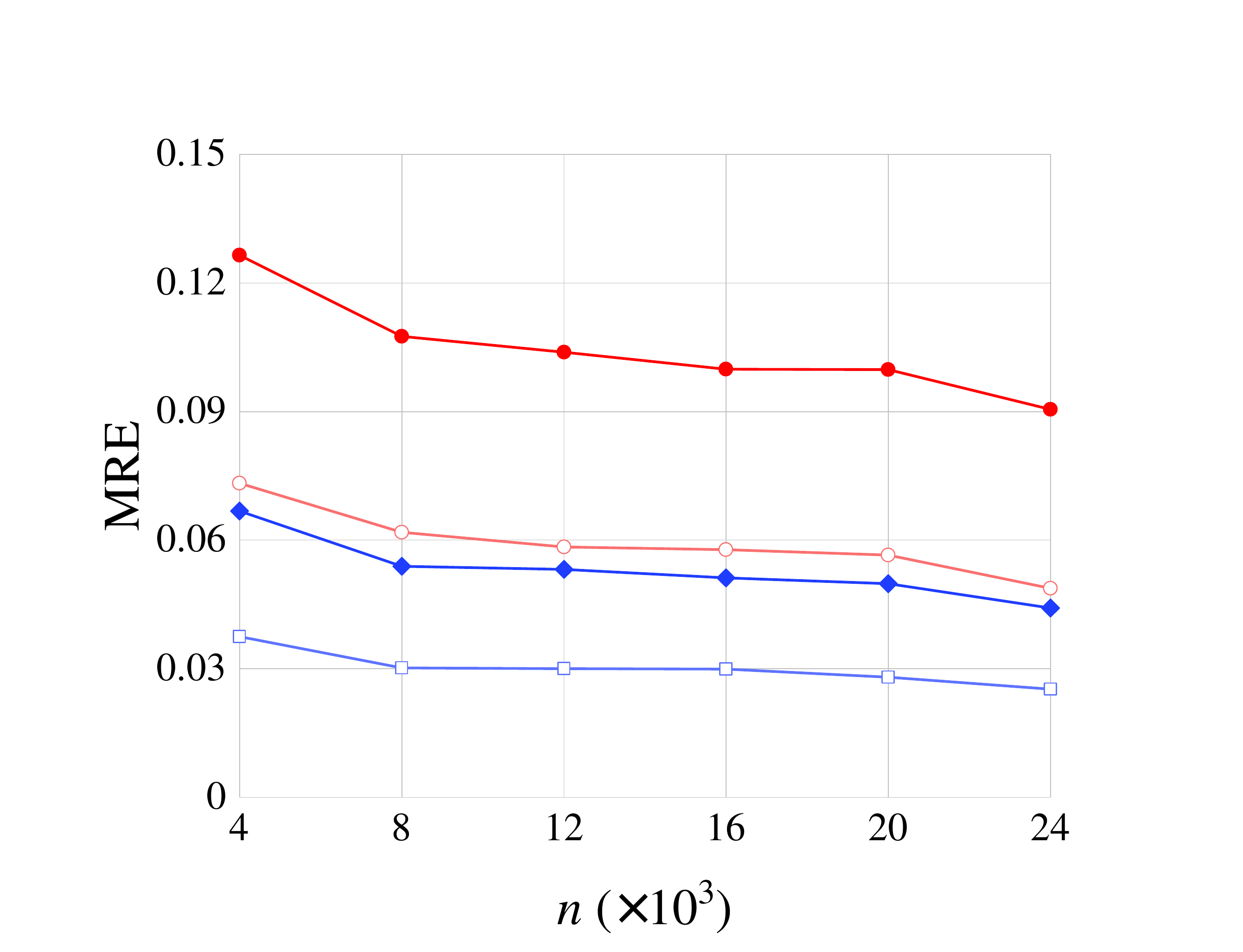}}
    \caption{Relation between $n$ and the MRE in $k$-stars counting }
 % \label{fig:data_distribution}
 % \vspace{0.2in}
\end{figure} 

\textbf{Summary of results.} In conclusion, a large number of experimental results show that the estimation error of subgraph counts can be greatly reduced under FGR-DP. As described in Section I, a unified protection strategy will not only overprotect the unimportant edges of the social graph, reducing the utility of graph analysis, but it will also cause issues such as insufficient protection of the core edges. 
We can provide fine-grained protection for edges with different privacy levels, which is the reason why our algorithm can significantly improve accuracy.

\section{Conclusions}\label{Conclusions}
This paper proposes a novel privacy definition called FGR-DP to provide fine-grained privacy graph analysis in decentralized social networks. Under FGR-DP, we design a privacy-preserving subgraph collection algorithm for $k$-stars counting and triangle counting, respectively, which can achieve better estimation accuracy with fine-grained privacy protection. Furthermore, 
we show how our algorithms are naturally expandable to multi-level privacy tasks. We then conduct comprehensive experiments on several real social graph datasets and show the superiority of the proposed algorithms. In the future, we will improve our algorithms to apply to more complex graph data analysis tasks.

\bibliographystyle{IEEEtran}
\bibliography{re.bib}
\end{document}